\documentclass[11pt,a4paper]{article}

\usepackage{amsthm,amsmath,amsfonts,amssymb}
\usepackage[margin=2cm]{geometry}
\usepackage{colonequals}
\usepackage{graphicx}
\usepackage{ifthen,comment}

\usepackage{xcolor}
\usepackage{tikz}
\usetikzlibrary{decorations.pathreplacing}
\usepackage{quantikz}
\usepackage{authblk}

\usepackage{algorithm}
\usepackage{float}
\usepackage{algpseudocode}

\usetikzlibrary{shapes.geometric}
\newboolean{ElectronicVersion}
\setboolean{ElectronicVersion}{true}

\ifthenelse{\boolean{ElectronicVersion}}{
    \usepackage[pdftex,bookmarks,pagebackref,
	plainpages=false, 
        pdfpagelabels=true 
        ]{hyperref}}{}

\makeatletter
\newtheorem*{rep@theorem}{\rep@title}
\newcommand{\newreptheorem}[2]{%
\newenvironment{rep#1}[1]{%
 \def\rep@title{#2 \ref*{##1}}%
 \begin{rep@theorem}}%
 {\end{rep@theorem}}}
\makeatother

\usepackage{hyperref}
\hypersetup{
    bookmarksnumbered=true, 
    unicode=false, 
    pdfstartview={FitH}, 
    pdftitle={A Quantum Algorithm for st-Transport on Flat Connection Graphs}, 
    pdfauthor={Stacey Jeffery, Tobias J. Osborne and Galina Pass}, 
    pdfsubject={}, 
    pdfcreator={}, 
    pdfproducer={}, 
    pdfkeywords={}, 
    pdfnewwindow=true, 
    colorlinks=true, 
    linkcolor=blue, 
    citecolor=blue, 
    filecolor=blue, 
    urlcolor=blue 
}

\usepackage{aliascnt}
\usepackage[nameinlink]{cleveref}
\Crefname{section}{Section}{Sections}

\newtheorem{theorem}{Theorem}[section]

\newaliascnt{definition}{theorem}
\newtheorem{definition}[definition]{Definition}
\aliascntresetthe{definition}
\Crefname{definition}{Definition}{Definitions}

\newaliascnt{lemma}{theorem}
\newtheorem{lemma}[lemma]{Lemma}
\aliascntresetthe{lemma}
\Crefname{lemma}{Lemma}{Lemmas}

\newaliascnt{proposition}{theorem}
\newtheorem{proposition}[proposition]{Proposition}
\aliascntresetthe{proposition}
\Crefname{proposition}{Proposition}{Propositions}

\newaliascnt{corollary}{theorem}
\newtheorem{corollary}[corollary]{Corollary}
\aliascntresetthe{corollary}
\Crefname{corollary}{Corollary}{Corollaries}

\newaliascnt{claim}{theorem}
\newtheorem{claim}[claim]{Claim}
\aliascntresetthe{claim}
\Crefname{claim}{Claim}{Claims}

\newaliascnt{example}{theorem}

\aliascntresetthe{example}
\Crefname{example}{Example}{Examples}

\newaliascnt{conjecture}{theorem}

\aliascntresetthe{conjecture}
\Crefname{conjecture}{Conjecture}{Conjectures}

\newaliascnt{aside}{theorem}

\aliascntresetthe{aside}
\Crefname{aside}{Aside}{Asides}

\newaliascnt{remark}{theorem}
\newtheorem{remark}[remark]{Remark}
\aliascntresetthe{remark}
\Crefname{remark}{Remark}{Remarks}

\newaliascnt{fact}{theorem}

\aliascntresetthe{fact}
\Crefname{fact}{Fact}{Facts}

\newaliascnt{problem}{theorem}
\newtheorem{problem}[problem]{Problem}
\aliascntresetthe{problem}
\Crefname{problem}{Problem}{Problems}

\newcommand{\eq}[1]{\hyperref[eq:#1]{(\ref*{eq:#1})}}
\renewcommand{\sec}[1]{\hyperref[sec:#1]{Section~\ref*{sec:#1}}}
\newcommand{\thm}[1]{\hyperref[thm:#1]{Theorem~\ref*{thm:#1}}}
\newcommand{\lem}[1]{\hyperref[lem:#1]{Lemma~\ref*{lem:#1}}}
\newcommand{\cor}[1]{\hyperref[cor:#1]{Corollary~\ref*{cor:#1}}}
\newcommand{\app}[1]{\hyperref[app:#1]{Appendix~\ref*{app:#1}}}
\newcommand{\tab}[1]{\hyperref[tab:#1]{Table~\ref*{tab:#1}}}
\newcommand{\defin}[1]{\hyperref[def:#1]{Definition~\ref*{def:#1}}}
\newcommand{\fig}[1]{\hyperref[fig:#1]{Figure~\ref*{fig:#1}}}
\newcommand{\clm}[1]{\hyperref[claim:#1]{Claim~\ref*{claim:#1}}}
\newcommand{\conj}[1]{\hyperref[conj:#1]{Conjecture~\ref*{conj:#1}}}
\newcommand{\rem}[1]{\hyperref[rem:#1]{Remark~\ref*{rem:#1}}}
\newcommand{\algo}[1]{\hyperref[algo:#1]{Algorithm~\ref*{algo:#1}}}
\newcommand{\probl}[1]{\hyperref[prob:#1]{Problem~\ref*{prob:#1}}}

\newcommand{\thmthm}[2]{\hyperref[thm:#1]{Theorem~\ref*{thm:#1}} and~\hyperref[thm:#2]{\ref*{thm:#2}}}
\newcommand{\lemlem}[2]{\hyperref[lem:#1]{Lemma~\ref*{lem:#1}} and~\hyperref[lem:#2]{\ref*{lem:#2}}}

\usepackage{color}
\definecolor{darkgreen}{rgb}{0,.5,0}
\definecolor{darkred}{rgb}{.7,.3,.3}
\definecolor{deepblue}{rgb}{0,.1,.7}

\def\ket#1{{\lvert}#1\rangle}
\def\bra#1{{\langle}#1\rvert}

\def\braket#1#2{{{\langle}#1\vert}#2\rangle}
\def\abs#1{\left| #1 \right|}

\def\norm#1{\left\| #1 \right\|}

\newcommand{\eps}{\varepsilon}

\def\tO{{\widetilde{O}}}

\def\tilt{{\hat t}}
\def\tils{{\hat s}}

\newcommand\restr[2]{{
  \left.\kern-\nulldelimiterspace 
  #1 
  \vphantom{\big|} 
  \right|_{#2} 
  }}

\title{A Quantum Algorithm for $st$-Transport on Flat Connection Graphs}
\author[1,2]{Stacey Jeffery}
\author[3]{Tobias J. Osborne}
\author[1]{Galina Pass}

\affil[1]{QuSoft \& CWI, Amsterdam}
\affil[2]{QLever \& University of Amsterdam}
\affil[3]{Institut für Theoretische Physik and L3S Research Center, Leibniz Universität Hannover}
\date{}

\begin{document}
\vskip10pt

\maketitle
\begin{abstract}
We study a generalization of undirected $st$-connectivity to graphs whose edges carry quantum operations. Let $G=(V,E)$ be an undirected graph on $n$ vertices in which each edge $\{u,v\}$ is labeled by a unitary $U_{uv}\in\mathbb{C}^{k\times k}$, with $U_{vu}=U_{uv}^\dagger$. We assume the labels form a \emph{flat} connection: the ordered product of labels along any path between a pair of vertices $u$ and $v$ is independent of the path. Equivalently, the connection is pure gauge, i.e., gauge-equivalent to the trivial connection; such graphs are exactly the consistent connection graphs of spectral graph theory and the noiseless instances of group synchronization. Consequently, whenever $s$ and $t$ are connected, transporting a state from $s$ to $t$ defines a unique unitary $U_s(t)$.
Given states $\ket{\psi_s},\ket{\psi_t}\in\mathbb{C}^k$ and an oracle that returns the neighbours of a vertex while coherently applying the corresponding edge unitaries, the \emph{$st$-transport problem} is to decide whether $s$ and $t$ are connected and, if so, to estimate the squared overlap between $U_s(t)\ket{\psi_s}$ and $\ket{\psi_t}$ to additive error $\eps$. When $k=1$ and all labels are trivial, this is exactly undirected $st$-connectivity. We give a bounded-error quantum algorithm for $st$-transport that runs in time $\widetilde{O}(n/\eps)$ and uses $O(\log n+\log k+\log(1/\eps))$ space. We do this by designing a transducer and applying a Metropolis-Hastings reweighting to the input graph. We also prove an $\Omega(n)$ quantum query lower bound that holds even when $s$ and $t$ are promised to be connected, so for constant $\eps$ our algorithm is optimal up to polylogarithmic factors.
\end{abstract}

\section{Introduction}
A \emph{unitary-labeled graph}, also referred to as a \emph{connection graph}~\cite{singer2012vector,bandeira2013cheeger}, is a graph in which each edge is equipped with a unitary operator acting on a fixed internal Hilbert space. Such graphs naturally encode how quantum states transform when transported between adjacent vertices.

The idea of decorating the edges of a graph with unitaries originates in
fundamental physics. In order to explain the confinement of quarks, Wilson
proposed regulating quantum chromodynamics on a discrete spacetime
lattice~\cite{wilson1974confinement}: here the gauge field is nothing other
than an assignment of a unitary parallel transporter to each edge, and all
gauge-invariant physical information is carried by the ordered products of
these unitaries along paths, the eponymous \emph{Wilson loops}. (For discrete
gauge groups this construction was foreshadowed by Wegner's generalised Ising
models~\cite{wegner1971duality}.) Thus was born \emph{lattice gauge
theory}~\cite{kogut1979introduction,creutz1983quarks}, which, particularly in
the Hamiltonian formulation of Kogut and
Susskind~\cite{kogut1975hamiltonian}, is precisely the study of quantum
states transported along the edges of a unitary-labeled graph.

A most successful tool in the study of lattice gauge theory has been the
computer. Monte Carlo sampling of the euclidean path integral, pioneered by
Creutz~\cite{creutz1980montecarlo}, supplied the first convincing numerical
evidence for the coexistence of confinement and asymptotic freedom, and three
decades of refinement culminated in the ab initio determination of the light
hadron spectrum~\cite{durr2008abinitio}, a landmark quantitative confirmation
of the standard model. This success is not, however, unqualified: the sign
problem afflicting such sampling is NP-hard~\cite{troyer2005computational},
and the regimes where importance sampling fails --- notably real-time
dynamics and finite fermion density --- remain largely beyond the reach of
classical computation.

These obstructions have made gauge theory a natural proving ground for
quantum computation. Quantum algorithms for simulating Kogut--Susskind
Hamiltonians were proposed early on~\cite{byrnes2006simulating}, the
real-time dynamics of a simple gauge theory has since been demonstrated on a
trapped-ion quantum computer~\cite{martinez2016realtime}, and tensor-network
and quantum-simulation approaches to lattice gauge theory now constitute a
thriving research programme~\cite{banuls2020simulating}. The connection to
complexity theory runs deeper than simulation: computing scattering
amplitudes in quantum field theory admits an efficient quantum
algorithm~\cite{jordan2012quantum} and is, indeed,
BQP-complete~\cite{jordan2018bqp}, so that quantum field theories capture the
full power of quantum computation. Unitary-labeled graphs distill the
structure underlying these developments --- unitary parallel transport along
the edges of a graph --- into its combinatorial core, where the tools of
quantum query and time complexity can be brought directly to bear.

Since its introduction, the idea of associating unitary operators with the edges of a graph has been further developed in a variety of settings. Such unitary-labeled graphs appear in work across physics, mathematics, and computer science as a way of augmenting classical graphs with additional local structure; see, e.g.,~\cite{robertson2025generalization} for a representative recent example and further references. From a modeling perspective, unitary-labeled graphs can be viewed as a natural extension of classical graphs, in which edges carry additional algebraic structure.
Related structures have also appeared in the quantum complexity literature, notably in~\cite{bausch2017complexity}.

\begin{remark}[Related notions] Flatness has been studied under other names. In spectral graph theory, graphs whose edges carry unitary (or orthogonal) labels are called \emph{connection graphs}, and flatness is called \emph{consistency}~\cite{singer2012vector,bandeira2013cheeger}. A unitary-labeled graph is flat exactly when one can assign a unitary $g_u$ to each vertex $u$ so that $U_{uv}=g_vg_u^\dagger$ for every edge $\{u,v\}$ (fix a vertex $r$ in each connected component and take $g_u=U_r(u)$); in physics terminology, such a connection is \emph{pure gauge}. Recovering the $g_u$ from the edge labels is the noiseless case of \emph{group synchronization}~\cite{singer2011angular}. If the labels were given as explicit $k\times k$ matrices, $st$-transport would be classically easy: find any $st$-path and multiply the labels along it. The interest of our setting is that the labels can only be applied as black boxes, the internal dimension $k$ may be exponentially large, and we want small space.
\end{remark}

\paragraph{The $st$-transport problem.} In this paper, we study a promise problem on an undirected graph in which each edge $\{u,v\}$ is equipped with a unitary $U_{uv}\in \mathbb{C}^{k\times k}$ satisfying $U_{vu}=U_{uv}^\dagger$. Under a \emph{flatness} promise, meaning that the ordered product of edge unitaries along any path depends only on its endpoints (equivalently, the product of unitaries around any cycle is the identity), any two vertices $s,t\in V(G)$ in the same connected component define a unitary $U_s(t)$. Given boundary states $\ket{\psi_s},\ket{\psi_t}\in\mathbb{C}^k$, the task is to estimate the overlap $\abs{\bra{\psi_t} U_s(t) \ket{\psi_s}}^2$ with additive error $\varepsilon$ if $s$ and $t$ are connected, or output that $s$ and $t$ are disconnected otherwise. We call this the \emph{$st$-transport problem} (see \probl{gauge_problem_estimate} for a formal definition). To approach this problem, we study a variant (\probl{gauge_problem_state_gen}) in which the task is to map $\ket{\psi_s}$ to $U_s(t)\ket{\psi_s}$.
These problems can be viewed as a generalization of the well-studied undirected $st$-connectivity problem, in which the task is to decide whether there exists a path between two vertices $s$ and $t$ in a graph, usually by starting from $s$ and finding $t$ (or failing to do so). In the general setting considered here, the flatness promise guarantees that any path from $s$ to $t$ yields the same global transformation. The algorithmic challenge is therefore to access the action of the induced map $U_s(t)$ on $\ket{\psi_s}$ using only local oracle access to the graph and the edge unitaries.

Quantum algorithms with simultaneously optimal time and space complexity have been found for undirected $st$-connectivity, both in the adjacency array and adjacency matrix models~\cite{belovs2012stcon, apers_et_al:LIPIcs.ESA.2023.10}. In both models, these algorithms use $O(\log n)$ space; the running time is $\widetilde{O}(n)$ in the adjacency array model and $\widetilde{O}(n^{3/2})$ in the adjacency matrix model.
Our model (see \Cref{prob:gauge_problem_estimate}) generalizes the adjacency array model.

\paragraph{Our contribution.} We show that the $st$-transport problem can be solved by a bounded-error quantum algorithm running in $\widetilde{O} \left( n / \varepsilon \right)$ time\footnote{Throughout this work, $\tO(\cdot)$ hides factors polylogarithmic in $n$, $k$ and $\frac{1}{\eps}$.} and using $O(\log n+\log k+\log(1/\eps))$ space on graphs with $n$ vertices, where $k$ is the dimension of the edge unitaries and $\varepsilon$ is the additive error. This shows that adding a $k$-dimensional internal degree of freedom under a flatness promise still admits a polynomial-time quantum algorithm using logarithmic space in the graph size and local dimension, for constant $\eps$. In fact, in the constant $\eps$ regime, our time and space upper bounds are both optimal -- we prove an $\Omega(n)$ query lower bound, and $\Omega(\log n+\log k)$ space is necessary for the oracle -- and match the complexity of the quantum algorithm for $st$-connectivity from~\cite{apers_et_al:LIPIcs.ESA.2023.10}.

On the way, we design a quantum algorithm that approximates the state $U_s(t)\ket{\psi_s}$ using similar resources (see \Cref{prob:gauge_problem_state_gen} for the precise problem definition).
We additionally prove an $\Omega(n)$ lower bound for $st$-transport, even under the promise that $s$ and $t$ are connected.

Throughout the paper, when we refer to the running time of an algorithm, we count the following resources:
(1) elementary quantum gates, i.e., unitary operations acting on at most a constant number of qubits;
(2) queries to the graph oracle $O_G$ through which we access the graph and the associated unitaries, and, for weighted graphs, to the weighted-neighbour oracle $O_w$ (see \Cref{prob:gauge_problem_estimate}); and
(3) the state-preparation unitaries $O_s$ and $O_t$ that prepare the boundary states $\ket{\psi_s}$ and $\ket{\psi_t}$ from $\ket{0}$.

\paragraph{Our techniques.} Our algorithms use the framework of transducers \cite{belovs2023LasVegasTime}. In particular, it was known that a quantum walk operator is a transducer, and can even implement the state transformation $\ket{s}\mapsto\ket{t}$ for classical labels $s$ and $t$ associated to a pair of vertices \cite{belovs2024weldedTrees}. We extend this quantum walk transducer to flat unitary-labeled graphs, in a manner similar to multidimensional quantum walks~\cite{jeffery2022kDist}. The resulting construction can be viewed as a generalization of the line graphs associated with quantum algorithms in~\cite{jeffery2022kDist}, \cite{jeffery2022subroutines} and~\cite{jeffery_et_al:LIPIcs.STACS.2025.54}. This gives a transducer for the action $\ket{s}\ket{\psi_s}\mapsto \ket{t}U_s(t)\ket{\psi_s}$, up to a known sign, whenever $s$ and $t$ are connected --- otherwise $\ket{s}\ket{\psi_s}$ is mapped to itself. This allows us to detect connectedness. 
To estimate the overlap of $\ket{\psi_t}$ with $U_s(t)\ket{\psi_s}$ in the connected case, we compose this transducer with an amplitude estimation transducer. Composing at the level of transducers avoids the additional precision costs of repeatedly invoking an approximate state-generation algorithm inside the estimation procedure. The composition has exactly the desired transduction action, and we convert it into a quantum algorithm only at the end, with constant conversion error. This yields overlap estimation to additive error $\varepsilon$ in time $\widetilde O (\sqrt{\mathbf{W}\mathbf{R}} / \varepsilon)$, where {$\mathbf{W}$} is an upper bound on the total weight of the graph, and {$\mathbf{R}$} an upper bound on the effective resistance between $s$ and $t$.

Finally, we extend the Metropolis-Hastings graph construction, used in \cite{apers_et_al:LIPIcs.ESA.2023.10} in the setting of a quantum walk for undirected $st$-connectivity, to preserve both flatness and the endpoint transport. The resulting graph has essentially the same edges (each now subdivided into two), but modified weights, and satisfies $W'R_{s,t}' = O (n^2)$, where $W'$ is the total weight, and $R'_{s,t}$ the effective resistance. This gives the claimed $\widetilde O(n / \varepsilon)$ running time. 

Our lower bound follows by encoding a parity instance in the unitary labels along a path, establishing an $\Omega(n)$ query lower bound.

\bigskip

The thesis~\cite{pass2026quantum} presents a preliminary version of this work using the framework of subspace graphs~\cite{jeffery_et_al:LIPIcs.STACS.2025.54}, giving an algorithm for deciding whether $\ket{\psi_t}=U_s(t)\ket{\psi_s}$ or the two are orthogonal.

\paragraph{Organization.}
The remainder of this paper is organized as follows. In~\Cref{sec:prelims} we give preliminaries: graph-theoretic preliminaries and problem definitions (\Cref{sec:gauge-problem}) including the Metropolis-Hastings reweighting (\Cref{sec:MH-graph}); and the transducer framework (\Cref{sec:transducers}). In \Cref{sec:gauge-state-preparation-transducer}, we give transducers and algorithms for generating the state $U_s(t)\ket{\psi_s}$, and estimating its overlap with $\ket{\psi_t}$ (the $st$-transport problem), in any weighted unitary-labeled graph, with complexities depending
on the total weight and effective resistance, with an upper bound in general of $\tO(n^{3/2}/\eps)$ for unweighted graphs. In \Cref{sec:MH-connection} we show how to implement a Metropolis-Hastings reweighting on a unitary-labeled graph, and in \Cref{sec:linear-alg} we use this to improve the upper bound to $\tO(n/\eps)$. Finally, in \Cref{sec:gauge-lower-bound} we prove the lower bound.

\subsection{Acknowledgments and AI use}

We thank Maris Ozols and Michael Walter for helpful comments and discussions on these ideas.

This work is co-funded by the European Union (ERC, ASC-Q, 101040624); Divide \& Quantum  (with project number 1389.20.241) of the research programme NWA-ORC, which is (partly) financed by the Dutch Research Council (NWO); and the Dutch National Growth Fund (NGF), as part of the QDNL programme.

LLMs were used in this work for some computations, proofs, and figures, as well as proofreading. All high-level ideas came from the authors. We have checked all proof details, and rewritten LLM-generated proofs as needed. We take full responsibility for the correctness and quality of writing.

\section{Preliminaries}\label{sec:prelims}

\subsection{\texorpdfstring{Graph notation and the $st$-transport problem}{Graph notation and the st-transport problem}}\label{sec:gauge-problem}

\paragraph{Graphs.} Let $G=(V,E)$ be an undirected graph on $n$ vertices with no self-loops. For each vertex $u\in V$, let $N(u)$ denote the set of neighbors of $u$, and let
$d_u=|N(u)|$ denote its degree.

We can augment the edges of $G$ with positive real weights, $w_{uv}=w_{vu}$, to get a weighted graph. We can extend $w$ to all pairs of vertices by letting $w_{uv}=0$ whenever $\{u,v\}\not\in E$. We let
$$w_u:=\sum_{v\in N(u)}w_{uv}$$
denote the total weight of any $u\in V$. Then the total weight of $G$ is defined:
$$W(G):=\sum_{u,v \in V} w_{uv} = \sum_{u \in V} w_u.$$
Note that we are using the convention that the total weight counts every edge weight \emph{twice}.

\paragraph{Graph access.} In problems in which the input is a graph, there are multiple possible ways of accessing the graph. The most natural for the setting of random walks is the \emph{adjacency array model}, where it is assumed that for each vertex $u$, there is a bijection
\[
    f_u:\{0,\ldots,d_u-1\}\to N(u)
\]
accessed via an oracle.
We say that $f_u(i)$ is the \emph{$i$-th neighbour of $u$}. Accessing a graph through oracles of this kind lends itself naturally to the implementation of a random walk, a step of which can be implemented by sampling an index $i\in\{0,\dots,d_u-1\}$, where $u$ is the current vertex, and then letting $f_u(i)$ be the next vertex.
By this setup, for every edge $\{u,v\}\in E$, there exist unique indices
$i_{u,v}\in\{0,\ldots,d_u-1\}$ and
$i_{v,u}\in\{0,\ldots,d_v-1\}$ such that
\[
f_u(i_{u,v})=v
\qquad\text{and}\qquad
f_v(i_{v,u})=u.
\]

In a \emph{quantum random walk}, this way of accessing an input graph is not quite sufficient, as the action $\ket{u,i}\mapsto \ket{f_u(i)}$ is not unitary, so it is standard to assume stronger access to the reversible map that acts, for any $\{u,v\}\in E$, as $\ket{u,i_{u,v}}\mapsto\ket{v,i_{v,u}}$, called a \emph{rotation map}. \Cref{fig:neighbour_ordering} shows how $\ket{u,i_{u,v}}$ and $\ket{v,i_{v,u}}$ are two different encodings of the same edge, and thus the rotation map is reversible.

\paragraph{Unitary-labeled graphs and $st$-transport.} A unitary-labeled graph is simply a graph where we decorate each edge with a unitary, as we make precise in the following definition.
\begin{definition}\label{def:labeled-graph}
A \emph{unitary-labeled graph}, or \emph{connection graph}, is an undirected weighted graph $G=(V,E,w)$ equipped with a unitary operator $U_{uv}\in\mathbb{C}^{k\times k}$ for each edge $\{u,v\}\in E$, satisfying $U_{vu}=U_{uv}^{\dagger}$. We say that a unitary-labeled graph is \emph{flat} (or \emph{consistent}) if for any $u, v \in V$ and any two paths
\[
u = u_0, u_1, \ldots, u_L = v
\quad \text{and} \quad
u = u'_0, u'_1, \ldots, u'_{L'} = v
\]
it holds that
\[
U_{u_{L-1}u_L} \cdots U_{u_0u_1}
= U_{u'_{L'-1}u'_{L'}} \cdots U_{u'_0u'_1} =: U_u(v).
\]
\end{definition}

When a unitary-labeled graph is given as input, we can access it via a generalization of the rotation map described above, that also encodes the unitaries, as black boxes that are applied to an arbitrary state while simultaneously rotating the edge, denoted by $O_G$ in the following problem definition.

\begin{problem}[$st$-transport]\label{prob:gauge_problem_estimate}
Let $G=(V,E,w)$ be a flat unitary-labeled graph (\defin{labeled-graph}) with two distinguished vertices $s, t \in V$.
Let $\ket{\psi_s}, \ket{\psi_t} \in \mathbb{C}^k$ be states associated with $s$ and $t$, respectively.
Given access to $G$ via an oracle $O_G$ acting, for all $u\in V$, $i\in \{0,\dots,d_u-1\}$, and $\ket{\psi}\in\mathbb{C}^k$ as
\[
O_G: \ket{u}\ket{i}\ket{\psi} \mapsto \ket{v}\ket{j}\,U_{uv}\ket{\psi},
\]
where $v=f_u(i)$, and $f_v(j)=u$, and acting as the identity on all other basis states, and oracle $O_w$ acting as
\[
O_w: \ket{u}\ket{0} \mapsto \frac{1}{\sqrt{w_u}} \sum_{i = 0}^{d_u - 1} \sqrt{w_{u, f_u (i)}} \ket{u}\ket{i},
\]
and the unitaries $O_s$ and $O_t$ that map $\ket{0}$ to $\ket{\psi_s}$ and $\ket{\psi_t}$ respectively, the task is to estimate the overlap
\[
\abs{\bra{\psi_t}U_s(t)\ket{\psi_s}}^2
\]
to additive error $\eps$ if $s$ and $t$ are connected, and output that $s$ and $t$ are disconnected otherwise.
\end{problem}
In our cost model, each call to the oracles $O_G$, $O_{w}$ and to the unitaries $O_s$ and $O_t$ is counted as unit time. We additionally assume that the weighted degrees $w_s, w_t$ of $s$ and $t$, an upper bound $\mathbf{W}$ on the total weight $W(G)$, and an upper bound $\mathbf{R}$ on the effective resistance ${\cal R}_{st}(G)$ (\Cref{def:resistance}) in the connected case are known in advance.

In order to solve \Cref{prob:gauge_problem_estimate}, we will first solve the following state generation problem, which may be of independent interest.
\begin{problem}[Transport state generation]\label{prob:gauge_problem_state_gen}
Let $G=(V,E,w)$ be a flat unitary-labeled graph (\defin{labeled-graph}) with two distinguished vertices $s, t \in V$, which are promised to be connected.
Let $\ket{\psi_s} \in \mathbb{C}^k$ be a state associated with $s$.
Given access to $G$ via oracles $O_G$ and $O_w$ as in \Cref{prob:gauge_problem_estimate},
and a unitary $O_s$ that maps $\ket{0}$ to $\ket{\psi_s}$, the task is to generate the state $U_s(t)\ket{\psi_s}$.
\end{problem}

\paragraph{Flows and effective resistance.} In order to analyze our quantum walk inspired transducer, we need to define the \emph{effective resistance}, introduced to the study of quantum walks in~\cite{belovs2013ElectricWalks} (published in~\cite{belovs2013TimeEfficientQW3Distintness}), for which we need the concept of a \emph{flow}:

\begin{definition}\label{def:flow}
Let $G=(V,E,w)$ be a weighted undirected graph. An oriented edge of $G$ is an ordered
pair $(u,v)$ such that $\{u,v\}\in E$. A flow on $G$ is a real-valued
function $\theta$ on oriented edges satisfying $\theta_{uv}=-\theta_{vu}$ for every $\{u,v\}\in E$. For $u\in V$, define the net flow out of $u$ by
\[
    \theta(u):=\sum_{v\sim u}\theta_{uv}.
\]
Let $s,t\in V$. We call $\theta$ a unit $st$-flow if
\[
    \theta(s)=1,\qquad \theta(t)=-1,\qquad \theta(u)=0
    \quad\text{for all }u\in V\setminus\{s,t\}.
\]
We call $\theta$ an optimal unit $st$-flow if it minimizes
\[
    \sum_{\{u,v\}\in E}\frac{\theta_{uv}^2}{w_{uv}},
\]
where either orientation of each edge may be used in the summand, since $\theta_{uv}^2=\theta_{vu}^2$.
\end{definition}

\begin{definition}\label{def:resistance}
    On a weighted undirected graph $G$  with $s,t \in V(G)$ that are connected, the \emph{effective resistance} is defined 
    $${\cal R}_{st}(G)\colonequals\min_{\mbox{unit $st$-flows }\theta}\sum_{\{u,v\} \in E}\frac{\theta_{uv}^2}{w_{uv}}.$$
\end{definition}

\subsection{Metropolis-Hastings graph}\label{sec:MH-graph}
In this subsection, we recall the construction of the Metropolis-Hastings graph used in~\cite{apers_et_al:LIPIcs.ESA.2023.10}, based on the Metropolis-type walk analyzed in~\cite{kosowski2013faster}. The construction replaces every edge of an unweighted graph by a path of length two and assigns weights depending on the degrees of its endpoints. This produces a weighted random walk with a uniformly quadratic bound on the hitting time. Later, we will equip the resulting weighted graph with unitary edge labels.
\begin{definition}\label{def:MH}
Let $G = (V, E)$ be an unweighted graph. The corresponding \emph{Metropolis-Hastings graph} is the weighted graph $G'=(V',E',w)$ defined as follows. For every $u \in V$, we include a corresponding vertex $x_u$ in $V'$. In addition, for every edge $\{u,v\}\in E$, we add a new vertex $x_{u,v}$ that splits the edge into two new edges. Formally:
\begin{align*}
    V' &= \{x_u:u\in V\}\cup\{x_{u,v}:\{u,v\}\in E,u<v\}\\
    E' &= \{\{x_u,x_{u,v}\}, \{x_{u,v},x_v\}: \{u,v\}\in E,u<v\}.
\end{align*}
For each edge $\{u,v\} \in E$ with $u < v$, we define the weights
\[
w_{x_u,x_{u,v}}=\frac{1}{d_u},
\qquad
w_{x_{u,v},x_v}=\frac{1}{d_v}.
\]
\end{definition}
The main property of this construction that we need is the following quadratic upper bound on its hitting time.
\begin{lemma}[\cite{kosowski2013faster}, Lemma 2 of {arXiv} v2]\label{lem:MH}
Let $G=(V, E)$ be any unweighted graph with $\abs{V} = n$, and $G'$ the corresponding weighted Metropolis-Hastings graph as in \defin{MH}.
For any $s,t \in V$ connected by a path, the hitting time ${\cal H}_{x_s x_t} (G')$ (the expected number of steps needed to reach $x_t$ in a random walk starting from $x_s$) is at most $18 n^2$.
\end{lemma}
To translate this hitting-time bound into a bound involving effective resistance, we recall the following commute-time identity.
\begin{lemma}[\cite{chandra1989electrical,belovs2013ElectricWalks}]\label{lem:commute_time}
Let $G=(V, E, w)$ be a weighted graph of total weight $W(G)$. Let $s, t \in V$ be two special vertices connected by a path and ${\cal R}_{st} (G)$ be the effective resistance. Then $W(G) {\cal R}_{st} (G) = {\cal H}_{st} (G) + {\cal H}_{ts} (G)$.
\end{lemma}
Combining the two lemmas above gives the bound on the resistance-weight product that we will use in the complexity of our algorithm.

\begin{corollary}\label{cor:MH_WR}
Let $G=(V,E)$ be an unweighted graph with $\abs{V} = n$, and let $G'$ be the corresponding weighted Metropolis-Hastings graph from \defin{MH}. Let $s,t \in V$ be two distinct vertices connected by a path. Then $$W(G')\mathcal{R}_{x_s x_t}(G')\leq 36n^2.$$ Moreover,
$W(G') \leq 2n$ and $\mathcal{R}_{x_s x_t}(G')\leq 18n.$
\end{corollary}
\begin{proof}
Let $V_s$ be the vertex set of the connected component of $G$ containing $s$ and $t$, let $n_s = \abs{V_s}$, and let $G'_s$ be the corresponding connected component of $G'$ containing $x_s$ and $x_t$. Since every vertex in $V_s$ is non-isolated, the definition of the edge weights gives
$$W(G'_s) = 2 \sum_{u \in V_s} \sum_{v \in N(u)} \frac{1}{d_u} = 2n_s.$$
Applying \lem{MH} to the component induced by $V_s$ gives
$$\mathcal{H}_{x_s x_t}(G'_s) + \mathcal{H}_{x_t x_s}(G'_s) \leq 36n_s^2.$$
Therefore, by \lem{commute_time},
$$W(G'_s) \mathcal{R}_{x_s x_t} (G'_s) \leq 36 n_s^2.$$
That is
$$2 n_s \mathcal{R}_{x_s x_t}(G'_s) \leq 36 n_s^2,$$
and hence
$$\mathcal{R}_{x_s x_t}(G') = \mathcal{R}_{x_s x_t} (G'_s) \leq 18n_s \leq 18n.$$
For the entire graph, isolated vertices contribute no weight, and thus
$$W(G') = 2 \sum_{\substack{u \in V\\d_u>0}}\sum_{v\in N(u)}\frac{1}{d_u}=2\abs{\{u\in V:d_u>0\}}\leq 2n.$$
Combining these two bounds yields
\begin{equation*}
    W(G') \mathcal{R}_{x_s\,x_t} (G') \leq (2n)(18n) = 36 n^2.
\end{equation*}
\end{proof}

\subsection{Transducers}\label{sec:transducers}

A transducer is a unitary $S$ acting on a space ${\cal H}\oplus {\cal L}$ decomposed into a \emph{public space} (or \emph{boundary space}) ${\cal H}$, and \emph{private space} ${\cal L}$. Such an object is interesting, first of all because of the following, which is from \cite[Theorem~3.1]{belovs2023LasVegasTime}.

\begin{theorem}\label{thm:transduction-action}
    Let $S$ be a unitary on ${\cal H}\oplus {\cal L}$. There is a unique unitary $U$ on ${\cal H}$, such that for each $\ket{\xi}\in {\cal H}$, there exists $\ket{\nu}\in {\cal L}$ such that $S(\ket{\xi}+\ket{\nu})=U\ket{\xi}+\ket{\nu}$. Moreover, for every $\ket{\xi}\in {\cal H}$ and $\ket{\nu}\in {\cal L}$, if $S(\ket{\xi}+\ket{\nu})-\ket{\nu}\in {\cal H}$, then $S(\ket{\xi}+\ket{\nu})=U\ket{\xi}+\ket{\nu}$.
\end{theorem}
We call $U$ the \emph{transduction action} of $S$ (on ${\cal H}$), and we write $S:\ket{\xi}\rightsquigarrow\ket{\tau}$, or say that $S$ \emph{transduces} $\ket{\xi}$ to $\ket{\tau}$, when $U\ket{\xi}=\ket{\tau}$.
We call the vector $\ket{\nu}$ in the above theorem, which depends on $\ket{\xi}$, a \emph{catalyst} for $\ket{\xi}$. This catalyst is not unique, though there is a unique smallest one, and the map $\ket{\xi}\mapsto\ket{\nu(\xi)}$ taking a state to its smallest catalyst is linear~\cite[Theorem~5.1(c)]{belovs2023LasVegasTime}. Note that a catalyst will not generally be a unit vector, even if $\ket{\xi}$ is.
{We write $W(S,\ket{\xi})$ for the squared norm of the smallest catalyst for $\ket{\xi}$, and $W(S)$ for the supremum of $W(S,\ket{\xi})$ over unit vectors $\ket{\xi}\in{\cal H}$, and call these the \emph{transduction complexity} of $S$ (on $\ket{\xi}$).}

When we implement transducers, we will need to apply operations to the ${\cal H}$ part or the ${\cal L}$ part of a state only. To this end, as in \cite{belovs2023LasVegasTime} we assume that membership in ${\cal H}$ is maintained in a specific qubit, called the \emph{privacy qubit}. 

We should think of $S$ as something we can easily implement, and $U$ as something we \emph{want} to be able to implement. It is not obvious, but it turns out that it is always possible to implement $U$ using sufficiently many calls to $S$, where sufficiently many is governed by $\norm{\ket{\nu}}^2$. The following is from \cite[Theorem~3.2~and~5.5]{belovs2023LasVegasTime}.

\begin{theorem}[\cite{belovs2023LasVegasTime}]\label{thm:transducer-implementation}
    Fix any integer $K$, and error parameter $\eps\in (0,1)$. There is a quantum algorithm ${\cal A}_K(S)$ that makes $K$ controlled calls to a black-box unitary $S$, and $O(K)$ other gates, such that if $S$ is a transducer as above, for any $\ket{\xi}\in {\cal H}$ with catalyst $\ket{\nu}$, if $K\geq \frac{4\norm{\ket{\nu}}^2}{\eps}$, then $\norm{{\cal A}_K(S)\ket{\xi}-U\ket{\xi}}^2\leq \eps$.
\end{theorem}

As an example, if $S$ is a quantum walk operator on a graph with connected boundaries $s$ and $t$ (see \cite{belovs2024weldedTrees}, or \Cref{sec:gauge-state-preparation-transducer}), it has transduction action
$\ket{s}\rightsquigarrow \ket{t}$.
In \Cref{sec:gauge-state-preparation-transducer} we generalize this to the setting of unitary-labeled graphs.

It is also possible to turn any quantum algorithm into a transducer. The reason we might want to do this is that in many cases, a transducer has no error, even while the corresponding algorithm does, and so transducers can be composed without worrying about errors. In the lemma below, we work out one explicit example of transducer composition that we will use. It is based on a transducer for amplitude estimation, which we explicitly construct based on the standard amplitude estimation algorithm~\cite{brassard2000amplitudeAmpEst}. In an independent work, an optimal transducer for amplitude amplification is designed from scratch~\cite{dubus2026optimalTransducers}.

\begin{lemma}\label{lem:amp-amp-comp}
    Let $\ket{\pi}\in {\cal H}$ be a unit vector, and let $O=2\ket{\pi}\bra{\pi}-I$. Let $\Pi$ be an orthogonal projector on ${\cal H}$, and let $p=\norm{\Pi\ket{\pi}}^2$. Suppose $S_O$ is a transducer on ${\cal H}\oplus {\cal L}$ with transduction action $O$ on its public space ${\cal H}$. For every integer $M\geq 1$, there is a transducer $S_M$ on ${\cal S}\otimes{\cal T}\otimes {({\cal H}\oplus{\cal L})}$, where ${\cal S}$ and ${\cal T}$ are $M$-dimensional registers, with public space $\mathrm{span}\{\ket{0}_{\cal S}\}\otimes{\cal T}\otimes{\cal H}$, such that
    $$S_M:\ket{0}_{\cal S}\ket{0}_{\cal T}\ket{\pi}_{\cal H}\rightsquigarrow \ket{0}_{\cal S}\ket{\eta}_{\cal TH},$$
    for a state $\ket{\eta}$ such that measuring the ${\cal T}$ register of $\ket{\eta}$ yields a value $z$ such that if $\tilde p = \sin^2(\pi z/M)$,
    $$\Pr\left[|\tilde p - p|\leq 6\pi \frac{\sqrt{p}}{M}+\frac{9\pi^2}{M^2}\right]\geq \frac{3}{4}.$$
    The transduction complexity is $W(S_M,\ket{0}\ket{0}\ket{\pi})\leq (M-1)(1+W(S_O))$,
    and $S_M$ can be implemented using one controlled call to each of $S_O$ and $(I-2\Pi)$, two Fourier transforms $F_M^{\pm 1}$ on ${\cal T}$, and $O(\log M)$ other elementary gates.
\end{lemma}
\begin{proof}
Let
$$G(O)=O(I-2\Pi)$$
be the Grover iterate, which acts on ${\cal H}$, for $O$ the reflection around the initial state $\ket{\pi}$.
Note that in the usual setting of quantum search, and related problems, the reflection $2\Pi-I$ around the ``marked'' space is the oracle. In our setting of interest, the difficult thing, for which we will need a subroutine, is reflecting around $\ket{\pi}$. This also happens in the setting of search if you let $\ket{\pi}$ be the initial state with an extra qubit in which the marked subspace is already indicated with a $\ket{1}$ -- in that case, $2\Pi-I$ is easy, and all the work is in implementing $2\ket{\pi}\bra{\pi}-I$.

We first define a transducer $\tilde{S}_M(O)$ on 
$${\cal S}\otimes {\cal T}\otimes {\cal H}=\mathrm{span}\{\ket{s,t}:s,t\in\{0,\dots,M-1\}\}\otimes {\cal H}.$$
Let ${\sf Inc}$ be the unitary that increments the first register, $\cal S$, modulo $M$, and define
$$\tilde{S}_M(O)={\sf Inc}\sum_{s=0}^{M-1}\sum_{t=s+1}^{M-1}\ket{s,t}\bra{s,t}\otimes G(O)+{\sf Inc}\sum_{s=0}^{M-1}\sum_{t=0}^{s}\ket{s,t}\bra{s,t}\otimes I.$$
This is just the transducer obtained from applying the algorithm-to-transducer theorem, {\cite[Theorem~10.3]{belovs2023LasVegasTime}}, to the quantum amplitude estimation algorithm of \cite{brassard2000amplitudeAmpEst}, but we construct it explicitly for completeness, and to show that it does not require quantum random access gates. Its public space is:
$${\cal H}_M=\mathrm{span}\{\ket{0}_{\cal S}\}\otimes {\cal T}\otimes {\cal H}.$$
By assumption, $S_O$ is a transducer with transduction action $O$, meaning that for every state $\ket{\phi}\in {\cal H}$, there exists some $\ket{{\nu}}\in{\cal L}$, the private space of $S_O$, such that $S_O(\ket{\phi}+\ket{{\nu}})=O\ket{\phi}+\ket{{\nu}}$. Thus, for every $s\in\{0,\dots,M-1\}$, let $\ket{{\nu}_s}\in {\cal L}$ be the smallest catalyst for the unit vector $(I-2\Pi)G(O)^s\ket{\pi}$, so that $\norm{\ket{\nu_s}}^2=W(S_O,(I-2\Pi)G(O)^s\ket{\pi})\leq W(S_O)$ and
{\begin{equation}\label{eq:vs-catalyst}
S_O((I-2\Pi)G(O)^s\ket{\pi}+\ket{{\nu}_s})=O(I-2\Pi)G(O)^s\ket{\pi}+\ket{{\nu}_s}.
\end{equation}}

Consider the transducer $\tilde{S}_M(S_O)$. This is not completely well defined, even as a unitary, since $S_O$ acts on a space larger than ${\cal H}$, the space acted on by $O$, so let us define precisely what we mean by $\tilde{S}_M(S_O)$. It is a unitary on the space ${\cal S}\otimes{\cal T}\otimes ({\cal H}\oplus {\cal L})$. Let $G(S_O)$ be the unitary acting on ${\cal H}\oplus {\cal L}$ defined:
$$G(S_O)=S_O \left((I-2\Pi)_{\cal H}\oplus I_{\cal L}\right).$$
Then we can define 
$$\tilde{S}_M=\tilde{S}_M(S_O)=\left({\sf Inc}_{\cal STH}\oplus I_{\cal STL}\right)\left(\sum_{s=0}^{M-1}\sum_{t=s+1}^{M-1}\ket{s,t}\bra{s,t}\otimes G(S_O)+\sum_{s=0}^{M-1}\sum_{t=0}^{s}\ket{s,t}\bra{s,t}\otimes I\right).$$
That is, we extend ${\sf Inc}$ and $I-2\Pi$ to ${\cal S}\otimes {\cal T}\otimes{\cal L}$ by having them act as the identity (operationally, by controlling them on the privacy qubit of $S_O$).
The public space of this transducer is also ${\cal H}_M$ (its private space is bigger though, since it acts on a bigger space overall). 

Then, using $\ket{{\nu}_s}$ from \eqref{eq:vs-catalyst}, define:
$$\ket{{\nu}'}=\sum_{s=1}^{M-1}\frac{1}{\sqrt{M}}\sum_{t=s+1}^{M-1}\ket{s,t} G(O)^s\ket{\pi}_{\cal H}+\sum_{s=1}^{M-1}\frac{1}{\sqrt{M}}\sum_{t=0}^{s}\ket{s,t} G(O)^t\ket{\pi}_{\cal H}+\sum_{s=0}^{M-1}\frac{1}{\sqrt{M}}\sum_{t=s+1}^{M-1}\ket{s,t}\ket{{\nu}_s}_{\cal L},$$
which is in the private space of $\tilde{S}_M=\tilde{S}_M(S_O)$, since it is orthogonal to ${\cal H}_M$.
We will show that this is a catalyst, specifically:
\begin{equation}\label{eq:aa-composed-action}
\tilde{S}_M(\ket{0}_{\cal S}F_M\ket{0}_{\cal T}\ket{\pi}+\ket{{\nu}'})=\ket{0}_{\cal S}\frac{1}{\sqrt{M}}\sum_{t=0}^{M-1}\ket{t}G(O)^t\ket{\pi}+\ket{{\nu}'}
\end{equation}
where $F_M$ is an $M$-dimensional Fourier transform. Note that:
\begin{align*}
    &\ket{0}F_M\ket{0}\ket{\pi}+\ket{{\nu}'}\\
    ={}& \sum_{s=0}^{M-1}\frac{1}{\sqrt{M}}\sum_{t=s+1}^{M-1}\ket{s,t} G(O)^s\ket{\pi}_{\cal H}+\sum_{s=0}^{M-1}\frac{1}{\sqrt{M}}\sum_{t=0}^{s}\ket{s,t} G(O)^t\ket{\pi}_{\cal H}+\sum_{s=0}^{M-1}\frac{1}{\sqrt{M}}\sum_{t=s+1}^{M-1}\ket{s,t}\ket{{\nu}_s}_{\cal L}\\
    ={}&\sum_{s=0}^{M-1}\frac{1}{\sqrt{M}}\sum_{t=s+1}^{M-1}\ket{s,t} (G(O)^s\ket{\pi}_{\cal H}+\ket{{\nu}_s}_{\cal L})+\sum_{s=0}^{M-1}\frac{1}{\sqrt{M}}\sum_{t=0}^{s}\ket{s,t} G(O)^t\ket{\pi}_{\cal H}
\end{align*}
and so, using ${\sf Inc}'$ as a shorthand for ${\sf Inc}_{\cal STH}\oplus I_{\cal STL}$:
\begin{align*}
    & \tilde{S}_M\left(\ket{0}F_M\ket{0}\ket{\pi}+\ket{{\nu}'}\right)\\
    ={}& {\sf Inc}'\left(\sum_{s=0}^{M-1}\frac{1}{\sqrt{M}}\sum_{t=s+1}^{M-1}\ket{s,t} G(S_O)(G(O)^s\ket{\pi}_{\cal H}+\ket{{\nu}_s}_{\cal L})+\sum_{s=0}^{M-1}\frac{1}{\sqrt{M}}\sum_{t=0}^{s}\ket{s,t} G(O)^t\ket{\pi}_{\cal H}\right)\\
    ={}& {\sf Inc}'\left(\sum_{s=0}^{M-1}\frac{1}{\sqrt{M}}\sum_{t=s+1}^{M-1}\ket{s,t} S_O((I-2\Pi)G(O)^s\ket{\pi}_{\cal H}+\ket{{\nu}_s}_{\cal L})+\sum_{s=0}^{M-1}\frac{1}{\sqrt{M}}\sum_{t=0}^{s}\ket{s,t} G(O)^t\ket{\pi}_{\cal H}\right),
\end{align*}
since $(I-2\Pi)$ extends trivially to ${\cal L}$. Then by {\Cref{eq:vs-catalyst}}, this is:
\begin{align*}
    ={}& {\sf Inc}'\left(\sum_{s=0}^{M-1}\frac{1}{\sqrt{M}}\sum_{t=s+1}^{M-1}\ket{s,t} (O(I-2\Pi)G(O)^s\ket{\pi}_{\cal H}+\ket{{\nu}_s}_{\cal L})+\sum_{s=0}^{M-1}\frac{1}{\sqrt{M}}\sum_{t=0}^{s}\ket{s,t} G(O)^t\ket{\pi}_{\cal H}\right)\\
    ={}& {\sf Inc}'\left(\sum_{s=0}^{M-1}\frac{1}{\sqrt{M}}\sum_{t=s+1}^{M-1}\ket{s,t} (G(O)^{s+1}\ket{\pi}_{\cal H}+\ket{{\nu}_s}_{\cal L})+\sum_{s=0}^{M-1}\frac{1}{\sqrt{M}}\sum_{t=0}^{s}\ket{s,t} G(O)^t\ket{\pi}_{\cal H}\right).
\end{align*}
Then since ${\sf Inc}'$ also acts trivially on ${\cal S}\otimes {\cal T}\otimes {\cal L}$, this is (writing $\ket{M}_{\cal S}$ for $\ket{0}_{\cal S}$ since ${\sf Inc}$ increments modulo $M$):
\begin{align*}
    ={}& \frac{1}{\sqrt{M}}\left(\sum_{s=0}^{M-1}\sum_{t=s+1}^{M-1}\ket{s+1,t} G(O)^{s+1}\ket{\pi}_{\cal H}
    +\sum_{s=0}^{M-1}\sum_{t=0}^{s}\ket{s+1,t} G(O)^t\ket{\pi}_{\cal H}
    +\sum_{s=0}^{M-1}\sum_{t=s+1}^{M-1}\ket{s,t}\ket{{\nu}_s}_{\cal L}\right)\\
    ={}& \frac{1}{\sqrt{M}}\left(\sum_{s=1}^{M}\sum_{t=s}^{M-1}\ket{s,t} G(O)^{s}\ket{\pi}+\sum_{s=1}^{M}\sum_{t=0}^{s-1}\ket{s,t} G(O)^t\ket{\pi}
    +\sum_{s=0}^{M-1}\sum_{t=s+1}^{M-1}\ket{s,t}\ket{{\nu}_s}\right)\\
    ={}& \frac{1}{\sqrt{M}}\Bigg(\sum_{s=1}^{M-1}\sum_{t=s}^{M-1}\ket{s,t} G(O)^{s}\ket{\pi}_{\cal H}+\sum_{s=1}^{M-1}\sum_{t=0}^{s-1}\ket{s,t} G(O)^t\ket{\pi}_{\cal H}
    +\ket{M}\sum_{t=0}^{M-1}\ket{t}G(O)^t\ket{\pi}_{\cal H}\\
    & \qquad\qquad +\sum_{s=0}^{M-1}\sum_{t=s+1}^{M-1}\ket{s,t}\ket{{\nu}_s}_{\cal L}\Bigg)
    = \frac{1}{\sqrt{M}}\ket{0}\sum_{t=0}^{M-1}\ket{t}G(O)^t\ket{\pi}+\ket{{\nu}'}.
\end{align*}
We have thus established~\Cref{eq:aa-composed-action}.

Define the final transducer as $S_M = F_M^{-1}\tilde{S}_M F_M$, where the Fourier transforms are applied to ${\cal T}$, so they map the private space to itself. Then we have:
\begin{align*}
    S_M(\ket{0}_{\cal S}\ket{0}_{\cal T}\ket{\pi}+F^{-1}_M\ket{{\nu}'}) &= F_M^{-1}\tilde{S}_M F_M(\ket{0}_{\cal S}\ket{0}_{\cal T}\ket{\pi}+F^{-1}_M\ket{{\nu}'})\\
    &=F_M^{-1}\tilde{S}_M (\ket{0}_{\cal S}F_M\ket{0}_{\cal T}\ket{\pi}+\ket{{\nu}'})\\
    &=F_M^{-1}\left(\frac{1}{\sqrt{M}}\ket{0}\sum_{t=0}^{M-1}\ket{t}G(O)^t\ket{\pi}+\ket{{\nu}'}\right) &\mbox{by \Cref{eq:aa-composed-action}}\\
    &={\ket{0}\underbrace{\frac{1}{\sqrt{M}}\sum_{t=0}^{M-1}F_M^{-1}\ket{t}G(O)^t\ket{\pi}}_{=:\ket{\eta}}}+F_M^{-1}\ket{{\nu}'}.
\end{align*}
The state $\ket{\eta}$ is the final state of phase estimation of $G(O)$, and so by \cite[Theorem~12]{brassard2000amplitudeAmpEst}, measuring the register ${\cal T}$ to obtain a random variable $z\in\{0,\dots,M-1\}$ yields an estimate $\tilde p = \sin^2\left(\frac{\pi z}{M}\right)$ such that for any $k\geq 2$,
$$\Pr\left[\abs{\tilde p - p} \leq 2\pi k\frac{\sqrt{p(1-p)}}{M}+k^2\frac{\pi^2}{M^2}\right]\geq 1-\frac{1}{2(k-1)}.$$
If we take $k=3$, we get the claim in the lemma statement.

Next, we upper bound the transduction complexity:
\begin{align*}
    W(S_M,\ket{0}\ket{0}\ket{\pi}) &\leq \norm{F_M^{-1}\ket{{\nu}'}}^2 = \norm{\ket{{\nu}'}}^2\\
    &= \sum_{s=1}^{M-1}\frac{1}{M}\sum_{t=s+1}^{M-1}\norm{G(O)^s\ket{\pi}}^2+\sum_{s=1}^{M-1}\frac{1}{M}\sum_{t=0}^s\norm{G(O)^t\ket{\pi}}^2+\sum_{s=0}^{M-1}\frac{1}{M}\sum_{t=s+1}^{M-1}\norm{\ket{{\nu}_s}}^2\\
    &\leq \sum_{s=1}^{M-1}\frac{1}{M}\sum_{t=s+1}^{M-1}1+\sum_{s=1}^{M-1}\frac{1}{M}\sum_{t=0}^s1+\sum_{s=0}^{M-1}\frac{1}{M}\sum_{t=s+1}^{M-1}W(S_O)\\
    &= \frac{1}{M}(M-1)M+\frac{W(S_O)}{M}\frac{M(M-1)}{2}\leq (M-1)(1+W(S_O)).
\end{align*}
Finally, we bound the cost of one application of $S_M=F_M^{-1}\tilde{S}_MF_M$. Besides the two Fourier transforms, $\tilde{S}_M$ consists of: computing the bit $[t>s]$ into an ancilla, which takes $O(\log M)$ gates; one application of $G(S_O)$ controlled on this ancilla, which uses a call to $S_O$ and a call to $I-2\Pi$ controlled on the privacy qubit of $S_O$; uncomputing the ancilla; and ${\sf Inc}$, controlled on the privacy qubit, which takes $O(\log M)$ gates. This gives the stated implementation cost.
\end{proof}

\section{Transducer for the state preparation problem}\label{sec:gauge-state-preparation-transducer}

In this section, we will give a quantum algorithm for \Cref{prob:gauge_problem_state_gen}, that will apply $U_s(t)$ to any state $\ket{\psi_s}$, and use it to solve \Cref{prob:gauge_problem_estimate} -- estimating $|\bra{\psi_t}U_s(t)\ket{\psi_s}|^2$ -- with bounded error. To this end, consider any unitary-labeled graph $G$, with terminals $s$ and $t$, where we assume $s\neq t$ (the case $s=t$ is trivial; see the proof of \Cref{thm:linear-alg}), and augment $G$ by adding, to each of $s$ and $t$, dangling ``boundary'' edges, each having just the single endpoint.
{We encode the new edge out of $s$ by $\ket{s,\bot}$, where $\bot$ is a reserved value of the index register that is distinct from every valid neighbour index (for instance, an extra flag qubit that is $1$ only on boundary edges); similarly, we encode the new edge out of $t$ by $\ket{t,\bot}$. Throughout the paper we write
\[
\ket{\tils}:=\ket{s}\ket{\bot}\qquad\mbox{and}\qquad \ket{\tilt}:=\ket{t}\ket{\bot}.
\]
Since $\bot$ is not a valid neighbour index, $O_G$ acts as the identity on $\ket{\tils}$ and $\ket{\tilt}$ (tensored with anything in $\mathbb{C}^k$). We assume that the states $\ket{\tils}$ and $\ket{\tilt}$ can be prepared, and reflected around, using $O(1)$ elementary gates; then the reflections around $\ket{\tils}\ket{0}$ and $\ket{\tilt}\ket{0}$, where $\ket{0}\in\mathbb{C}^k$, cost $O(\log k)$ elementary gates, for checking that the last register is $\ket{0}$.}
See \Cref{fig:neighbour_ordering}.

\begin{figure}
\centering
\includegraphics[width=0.5\textwidth]{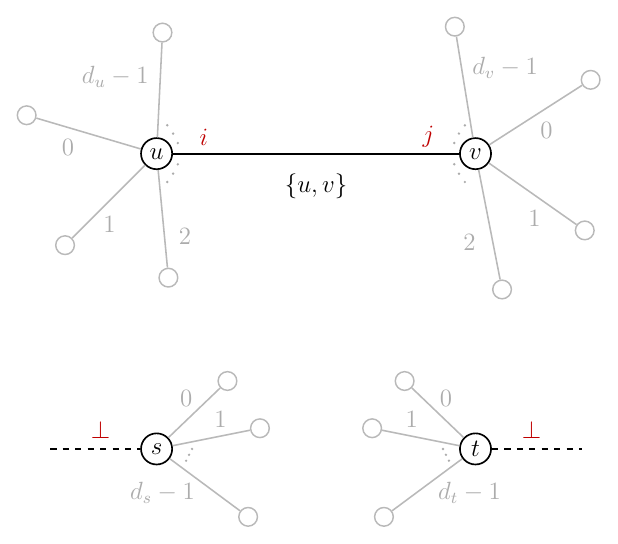}
\caption{Local indexing of neighbours in a graph. For an ordinary edge $\{u,v\}$, the same edge has two local indices: $i$ at $u$ and $j$ at $v$, meaning that $f_u(i)=v$ and $f_v(j)=u$; and $i_{u,v}=i$ and $i_{v,u}=j$. The gray edges indicate the other incident edges, whose indices range from $0$ to $d_u-1$ at $u$, and from $0$ to $d_v-1$ at $v$, excluding $i$ and $j$ correspondingly. For a boundary vertex $b \in \{s,t\}$, the ordinary graph edges are indexed by $0,\ldots,d_b-1$, while the dashed boundary edge has the special index $\bot$.}\label{fig:neighbour_ordering}
\end{figure}

In the remainder of this section, we will prove the following theorem.
\begin{theorem}\label{thm:gauge-state-preparation-transducer}
Let $\mathbf{W}\geq W(G)$ and $\mathbf{R}\geq {\cal R}_{st}(G)$ be known upper bounds. There exists a transducer $U_{\cal AB}=U_{\cal AB}(O_G,O_w)$ that depends on oracles $O_G,O_w$ as in \Cref{prob:gauge_problem_estimate}, such that:
\begin{enumerate}
\item the public space of $U_{\cal AB}$ is ${\cal H}=\mathrm{span}\{\ket{\tils},\ket{\tilt}\}\otimes\mathbb{C}^k$;

\item let
$\Sigma$ denote the transduction action of $U_{\cal AB}$ on ${\cal H}$.
If $s$ and $t$ are connected, then
\[
\Sigma \;=\; -\Big({\ket{\tilt}\bra{\tils}}\otimes U_s(t)
\;+\;{\ket{\tils}\bra{\tilt}}\otimes U_s(t)^\dagger\Big),
\]
and if $s$ and $t$ are not connected, then $\Sigma=-I_{\cal H}$. In both cases
$\Sigma=\Sigma^\dagger=\Sigma^{-1}$, and 
$W(U_{\cal AB})\leq 2\sqrt{\mathbf W\mathbf R}$;
\item $U_{\cal AB}$ can be implemented in $O(1)$ calls to $O_G$ and $O_w$, and $O(\log n)$ additional elementary operations.

\end{enumerate}
\end{theorem}
An immediate corollary of this theorem, and \Cref{thm:transducer-implementation}, is the following.

\begin{corollary}\label{cor:gauge-state-preparation-algorithm}
For any $\eps \in (0,1)$, there is a quantum algorithm that solves \Cref{prob:gauge_problem_state_gen} with error $\eps$, meaning that, if $s$ and $t$ are connected, it outputs a state $\ket{\varphi}$ such that
$$\norm{\ket{\varphi}-\ket{0}\ket{0}U_s(t)\ket{\psi_s}}^2\leq \eps,$$
using $O\left(\sqrt{\mathbf{W}\mathbf{R}}/\eps\right)$ calls to $O_G$ and $O_w$, one call to $O_s$, and $\widetilde{O}\left(\sqrt{\mathbf{W}\mathbf{R}}/\eps\right)$ additional elementary operations.
\end{corollary}
\begin{proof}
Prepare the state $\ket{\tils} O_s\ket{0} = \ket{\tils}\ket{\psi_s}$ using one call to $O_s$ and $O(\log n)$ additional elementary gates.
Apply \Cref{thm:transducer-implementation} to $S = U_{\cal AB}$ with error parameter $\eps$ and
$K = \left\lceil \frac{8 \sqrt{\mathbf W \mathbf R}}{\eps} \right\rceil$. 
Indeed, item~2 of \Cref{thm:gauge-state-preparation-transducer} gives $W(U_{\cal AB}) \leq 2 \sqrt{\mathbf W \mathbf R}$, so $K \geq 4 W(U_{\cal AB}) / \eps$, as required. Let $\ket{\tilde\varphi}$ be the resulting state. If $s$ and $t$ are connected, then item~2 of \Cref{thm:gauge-state-preparation-transducer} and \Cref{thm:transducer-implementation} give
$\norm{ \ket{\tilde\varphi} + \ket{\tilt}U_s(t)\ket{\psi_s} }^2 \leq \eps$. Let {$V_t$} be a unitary satisfying
${V_t}\ket{\tilt}\ket{\psi}=\ket{0}\ket{0}\ket{\psi}$
for every $\ket{\psi}\in\mathbb C^k$. This unitary can be implemented using $O(\log n)$ elementary gates. Output
$\ket{\varphi}:=-{V_t}\ket{\tilde\varphi}$. Then
$$\norm{\ket{\varphi}-\ket{0}\ket{0}U_s(t)\ket{\psi_s}}^2=\norm{-V_t(\ket{\tilde\varphi}+\ket{\hat{t}}U_s(t)\ket{\psi_s})}^2=\norm{\ket{\tilde\varphi}+\ket{\tilt}U_s(t)\ket{\psi_s}}^2\leq\eps.$$
The algorithm makes $K=O\left(\sqrt{\mathbf W\mathbf R}/\eps\right)$ controlled calls to $U_{\cal AB}$. By item~3 of \Cref{thm:gauge-state-preparation-transducer}, each such call uses $O(1)$ calls to $O_G$ and $O_w$ and $O(\log n)$ additional elementary operations. This gives the claimed complexity.
\end{proof}

A further corollary is a solution to the overlap problem. Naively, we could take the algorithm from {\Cref{cor:gauge-state-preparation-algorithm}}, and do amplitude estimation to estimate the amplitude on $\ket{\tilt}\ket{\psi_t}$, to get an algorithm for \Cref{prob:gauge_problem_estimate}, but, as discussed in \Cref{rem:suboptimal-alg}, this would give complexity $\widetilde{O}(\sqrt{\mathbf{WR}}/\eps^3)$ at best, because we need to put a lot of algorithmic effort into making the errors in the algorithm from {\Cref{cor:gauge-state-preparation-algorithm}} small enough to not impact our estimation. We can do much better by composing \emph{transducers}, where we don't have error, and only turning the final transducer into an algorithm when all composition is done. This final transformation does introduce errors, but only once, at the end. This is the tactic we use in the proof of \Cref{cor:gauge-state-overlap-algorithm} below, and we take this opportunity to emphasize this as a general principle: compositions should generally be done on transducers rather than algorithms.

\begin{remark}\label{rem:suboptimal-alg}
It is instructive to see why combining the state-generation algorithm of
\Cref{cor:gauge-state-preparation-algorithm} with standard estimation techniques
in a black-box manner is suboptimal. The simplest approach is repeated prepare-and-measure: prepare an approximation of ${\ket{\tilt}}U_s(t)\ket{\psi_s}$,
apply $O_t^\dagger$ to the last register, and measure, so that the probability
of the outcome ${\ket{\tilt}}\ket{0}$ is, up to a bias controlled by the
preparation error $\eps$, $p=|\bra{\psi_t}U_s(t)\ket{\psi_s}|^2$. Estimating a
Bernoulli probability to additive error $\eps$ requires $\Theta(1/\eps^2)$
samples, and each sample must be prepared to norm error $O(\eps)$ so that the
bias does not dominate. However, the preparation error also enters
quadratically: converting the transducer $U_{\cal AB}$ into a circuit with
norm error $\eps'$ requires $K=\Theta(\norm{\ket{\nu}}^2/\eps'^2)$ calls to
$U_{\cal AB}$ (\Cref{thm:transducer-implementation}), i.e.\
$\widetilde{O}(\sqrt{\mathbf{W}\mathbf{R}}/\eps^2)$ per sample. In total, the
precision is paid for twice --- once in the conversion and once in the
sampling --- yielding $\widetilde{O}(\sqrt{\mathbf{W}\mathbf{R}}/\eps^4)$.

Amplitude estimation helps, reducing the number of uses of
the preparation circuit to $O(1/\eps)$, but these uses are coherent, and
norm errors accumulate linearly across coherent calls. For the final state to
be within constant distance of the ideal one, the preparation must therefore
be accurate to norm error $O(\eps)$, at cost
$\widetilde{O}(\sqrt{\mathbf{W}\mathbf{R}}/\eps^2)$ per call, for a total of
$\widetilde{O}(\sqrt{\mathbf{W}\mathbf{R}}/\eps^3)$. The source of the loss is
the same in both cases: the conversion from transducer to algorithm is
performed inside the estimation loop, so its error parameter is coupled to the
target precision.
\end{remark}

Composing amplitude estimation
with $U_{\cal AB}$ \emph{at the level of transducers} avoids the problem in \Cref{rem:suboptimal-alg} entirely:
the composed transducer is exact, error is introduced only once --- in the
final conversion, with a constant error parameter --- and the resulting
complexity is $\widetilde{O}(\sqrt{\mathbf{W}\mathbf{R}}/\eps)$, as the following theorem makes precise.

\begin{theorem}\label{cor:gauge-state-overlap-algorithm}
There is a quantum algorithm that solves \Cref{prob:gauge_problem_estimate}
with bounded error, meaning that, with probability at least $2/3$, it correctly
reports that $s$ and $t$ are disconnected, or outputs an estimate $\tilde p$ such that
$\abs{\,\abs{\bra{\psi_t}U_s(t)\ket{\psi_s}}^2-\tilde p\,}\leq \eps$, using
$\widetilde{O}\left(\sqrt{\mathbf{W}\mathbf{R}}/\eps\right)$ calls to $O_G$,
$O_w$, $O_s$ and $O_t$, and
$\widetilde{O}\left(\sqrt{\mathbf{W}\mathbf{R}}/\eps\right)$ additional
elementary operations.
\end{theorem}
We prove this theorem in \Cref{sec:proof-main-theorem}.

In the next two sections, we prove \Cref{thm:gauge-state-preparation-transducer}. \Cref{sec:transducer} describes the transducer, and analyzes its action and transduction complexity, and \Cref{sec:implementation} describes how to implement the transducer in the claimed complexity.

\subsection{The transducer}\label{sec:transducer}

\paragraph{Public and private spaces.} We define the private (internal) space as the span of all edges in $G$, tensored with the extra space $\mathbb{C}^k$. Note that we have two ways of representing any edge $\{u,v\}\in E$: as $\ket{u,i_{u,v}}$ (from $u$'s perspective) or $\ket{v,i_{v,u}}$ (from $v$'s perspective):
\[
\mathcal{L} :=\left(
\bigoplus_{\{u,v\} \in E} \mathrm{span}\{ \ket{u,i_{u,v}}, \ket{v,i_{v,u}}\}\right)\otimes \mathbb{C}^k.
\]
We define the public (boundary) space as the span of the two new boundary edges we have conceptually added to $G$ at $s$ and $t$, tensored with the extra space $\mathbb{C}^k$:
\[
\mathcal{H} = \mathrm{span} \{ {\ket{\tils}}, {\ket{\tilt}}\}\otimes \mathbb{C}^k,
\]
from which we immediately see ${\ket{\tils}}\ket{\psi_s},{\ket{\tilt}}U_s(t)\ket{\psi_s}\in {\cal H}$.

\begin{remark}[Quantum walk operator]
Our transducer $U_{\cal AB}$ will be an extension of a quantum walk operator, which works roughly as follows (the expert reader may skip this remark). We can either assume $G$ is bipartite, or force it to be, by conceptually inserting a new vertex into each edge, which we do here. Let ${\cal A}$ be the half of the bipartition consisting of the vertices inserted into each edge, and let ${\cal B}$ be the original vertices. We can define a reflection for each of these:
around the span of \emph{star states} $\sum_{i=0}^{d_u-1}\sqrt{w_{u,f_u(i)}}\ket{u,i}$ for $u\in {\cal B}$; and similarly for ${\cal A}$, but since vertices in ${\cal A}$ are inserted in the middle of an edge $\{u,v\}$ of $G$, they have two incident edges of the same weight, so their star states have the form $\ket{u,i_{u,v}}+\ket{v,i_{v,u}}$. Note that each of these two sets of star states (for ${\cal A}$, or for ${\cal B}$) is pairwise orthogonal, which facilitates implementing this reflection -- this is the reason for introducing the bipartiteness condition.
For our construction, we will extend this notion, to incorporate the unitaries labeling the edges of $G$. For a more detailed exposition, see, for example,~\cite{jeffery2022kDist} or~\cite{jeffery2026QWlecture}.
\end{remark}

\paragraph{Transducer as a product of two reflections.} We define the transducer as a product of two reflections, around spaces ${\cal A}$ and ${\cal B}$ defined as follows. The space ${\cal A}$ contains a state for each edge of $G$, incorporating the application of the corresponding unitary:
\[
\mathcal{A} = \bigoplus_{\{u,v\} \in E} \mathrm{span} \left\{
\ket{u,i_{u,v}} \ket{z} + \ket{v,i_{v,u}}U_{uv}\ket{z} : z\in[k] \right\}.
\]
Importantly, as we will see shortly, the given oracle $O_G$ facilitates reflecting around this space. The space ${\cal B}$ contains a state for each vertex of $G$:
\begin{align*}
\mathcal{B} &=
\left(\bigoplus_{u \in V \setminus \{s,t\}} \mathrm{span} \left\{ \sum_{i=0}^{d_u - 1} \sqrt{w_{u, f_u(i)}} \ket{u,i} \right\}
\oplus \bigoplus_{b \in \{s,t\}} \mathrm{span} \left \{
\sqrt{w_0}\ket{b}\ket{\bot} + \sum_{i=0}^{d_b - 1} \sqrt{w_{b, f_b(i)}} \ket{b,i}
\right\}\right)\otimes \mathbb{C}^k,
\end{align*}
for some parameter $w_0 > 0$ to be chosen later.
The states on the left of the tensor are referred to in the literature as \emph{star} states, because they can be visualized as a star consisting of the edges coming out of a vertex $u$, proportional to their edge weights. For the boundary vertices $s$ and $t$, we have also included the boundary edges. Importantly, as we will soon see, the given oracle $O_w$ facilitates reflecting around this space.

Let $\Pi_{\cal A}$ and $\Pi_{\cal B}$ be the orthogonal projectors onto $\cal A$ and $\cal B$.
In the following lemma, we show that
$$U_{\mathcal{AB}}:=(2\Pi_{\cal A}-I)(2\Pi_{\cal B}-I)$$
has the desired transduction action.
\begin{lemma}\label{lem:gauge-state-preparation-transducer}
If $s$ and $t$ are connected, then for any unit vector $\ket{\psi_s}\in\mathbb{C}^k$, $U_{\cal AB}$ transduces
\[
\ket{\tils}\ket{\psi_s}\rightsquigarrow -\ket{\tilt}U_s(t)\ket{\psi_s}\qquad\mbox{and}\qquad \ket{\tilt}U_s(t)\ket{\psi_s}\rightsquigarrow-\ket{\tils}\ket{\psi_s},
\]
in each case with a catalyst $\ket{\nu}\in{\cal L}$ satisfying $\norm{\ket{\nu}}^2\leq \frac{1}{4}\left(\frac{W(G)}{w_0} + 2w_0{\cal R}_{st}(G) \right)$.
\end{lemma}
Thus, if we take $w_0=\sqrt{\mathbf{W}/\mathbf{R}}$ for some known upper bounds $\mathbf{W} \geq W(G)$ and $\mathbf{R} \geq {\cal R}_{st}(G)$, we get
$$\frac{W(G)}{w_0}+2 w_0{\cal R}_{st}(G)\leq \mathbf{W}\sqrt{\frac{\mathbf{R}}{\mathbf{W}}} + 2\mathbf{R}\sqrt{\frac{\mathbf{W}}{\mathbf{R}}} = 3\sqrt{\mathbf{R}\cdot\mathbf{W}}.$$
\begin{proof}
Assume $s$ and $t$ are connected. Then without loss of generality, we can assume $G$ is connected, by simply restricting our attention to the connected component that contains both $s$ and $t$. This restriction is without loss of generality because the spaces ${\cal A}$ and ${\cal B}$ decompose as orthogonal direct sums over the connected components of $G$, while all catalysts constructed below are supported entirely on the component containing $s$ and $t$. Moreover, restricting to this component can only decrease the total weight.

Let $\theta$ be an optimal unit $st$-flow (see \Cref{def:flow}). Define
\begin{align*}
\ket{\nu_0} &:= \frac{1}{\sqrt{w_0}}\sum_{\{u, v\} \in E} \sqrt{w_{uv}} \left(
\ket{u, i_{u,v}} U_s(u) \ket{\psi_s} + \ket{v, i_{v,u}}U_s(v)\ket{\psi_s} \right) = \sum_{u\in V}\sum_{i=0}^{d_u-1}\sqrt{\frac{w_{u,f_u(i)}}{w_0}}\ket{u,i}U_s(u)\ket{\psi_s},\\
\ket{\nu_1} &:= \sqrt{w_0}\!\!\!\sum_{\{u,v\} \in E} \frac{\theta_{uv}}{\sqrt{w_{uv}}} \left( \ket{u, i_{u,v}}U_s(u)\ket{\psi_s} - \ket{v,i_{v,u}}U_s(v)\ket{\psi_s} \right)=\sum_{u\in V}\sum_{i=0}^{d_u-1}\frac{\sqrt{w_0}\theta_{u,f_u(i)}}{\sqrt{w_{u,f_u(i)}}}\ket{u,i}U_s(u)\ket{\psi_s},
\end{align*}
where in the sums over unordered edges $\{u,v\}\in E$, we choose either orientation, $(u,v)$ or $(v,u)$, the choice of which does not impact the definition of either $\ket{\nu_0}$, or $\ket{\nu_1}$. This follows easily from $w_{uv}=w_{vu}$ for $\ket{\nu_0}$. For $\ket{\nu_1}$, if we reverse the orientation, then $\theta_{uv}$ changes sign, and the parenthesized vector also changes sign. Hence the product is unchanged.

We will show that
\begin{align}
    U_{\mathcal{AB}} \left( {\ket{\tils}}\ket{\psi_s} + {\ket{\tilt}} U_s(t)\ket{\psi_s} + \ket{\nu_0} \right) &=  - {\ket{\tils}}\ket{\psi_s} -  {\ket{\tilt}} U_s(t)\ket{\psi_s} + \ket{\nu_0}\label{eq:transduced_sums1}\\
    U_{\mathcal{AB}} \left( -{\ket{\tils}}\ket{\psi_s} + {\ket{\tilt}} U_s(t)\ket{\psi_s} + \ket{\nu_1} \right) &= -{\ket{\tils}}\ket{\psi_s} + {\ket{\tilt}} U_s(t)\ket{\psi_s} + \ket{\nu_1}.\label{eq:transduced_sums2}
\end{align}
Subtracting these two equations, it follows that
\begin{align*}
    U_{\cal AB}\left(2{\ket{\tils}}\ket{\psi_s}+\ket{\nu_0}-\ket{\nu_1}\right) &= -2{\ket{\tilt}}U_s(t)\ket{\psi_s}+\ket{\nu_0}-\ket{\nu_1}\\
    U_{\cal AB}\left({\ket{\tils}}\ket{\psi_s}+\frac{1}{2}(\ket{\nu_0}-\ket{\nu_1})\right) &= -{\ket{\tilt}}U_s(t)\ket{\psi_s}+\underbrace{\frac{1}{2}(\ket{\nu_0}-\ket{\nu_1})}_{=:\ket{\nu_-}},
\end{align*}
and adding them, it follows similarly that
\[
    U_{\cal AB}\left({\ket{\tilt}}U_s(t)\ket{\psi_s}+\frac{1}{2}(\ket{\nu_0}+\ket{\nu_1})\right) = -{\ket{\tils}}\ket{\psi_s}+\underbrace{\frac{1}{2}(\ket{\nu_0}+\ket{\nu_1})}_{=:\ket{\nu_+}}.
\]
Thus, it remains only to prove \Cref{eq:transduced_sums1} and \Cref{eq:transduced_sums2}, and upper bound the catalyst size $\norm{\ket{\nu_\pm}}^2 = \norm{\frac{1}{2}(\ket{\nu_0} \pm \ket{\nu_1})}^2$.
\paragraph{Proof of \Cref{eq:transduced_sums1}.} In order to prove \Cref{eq:transduced_sums1}, we first show that
\[
{\ket{\tils}}\ket{\psi_s} + {\ket{\tilt}} U_s(t)\ket{\psi_s} + \ket{\nu_0} \in \mathcal{B}.
\]
Then the edge terms in $\ket{\nu_0}$ can be grouped by vertices to get:
\begin{align*}
&\sqrt{w_0}\left({\ket{\tils}}\ket{\psi_s} + {\ket{\tilt}}U_s(t)\ket{\psi_s} + \ket{\nu_0}\right)\\
={}& \left(\sqrt{w_0} {\ket{\tils}}\ket{\psi_s} + \sum_{i=0}^{d_s - 1} \sqrt{w_{s,f_s(i)}} \ket{s,i} \ket{\psi_s} \right)
+ \left(\sqrt{w_0} {\ket{\tilt}}U_s(t)\ket{\psi_s} + \sum_{i=0}^{d_t - 1} \sqrt{w_{t,f_t(i)}} \ket{t,i} U_s(t)\ket{\psi_s} \right)\\
&+ \sum_{u \in V \setminus\{s,t\}} \left( \sum_{i=0}^{d_u-1} \sqrt{w_{u,f_u(i)}} \ket{u,i} U_s(u) \ket{\psi_s} \right)\\
={}& \left(\sqrt{w_0} {\ket{\tils}}+ \sum_{i=0}^{d_s - 1} \sqrt{w_{s,f_s(i)}} \ket{s,i}\right)\ket{\psi_s}
+ \left(\sqrt{w_0} {\ket{\tilt}} + \sum_{i=0}^{d_t - 1} \sqrt{w_{t,f_t(i)}} \ket{t,i} \right)U_s(t)\ket{\psi_s}\\
&+ \sum_{u \in V \setminus\{s,t\}} \left( \sum_{i=0}^{d_u-1} \sqrt{w_{u,f_u(i)}} \ket{u,i}  \right)U_s(u) \ket{\psi_s},
\end{align*}
which is clearly in ${\cal B}$.
Hence, the reflection through $\mathcal{B}$ leaves this vector unchanged.

Next, observe that since the graph is flat, for every edge $\{u,v\} \in E$, $U_s(v)=U_{uv}U_s(u)$.
Thus, for every edge $\{u,v\}\in E$,
\[
\ket{u,i_{u,v}}U_s(u)\ket{\psi_s} + \ket{v,i_{v,u}}U_s(v)\ket{\psi_s} = \ket{u,i_{u,v}}U_s(u)\ket{\psi_s} + \ket{v,i_{v,u}}U_{uv}U_s(u)\ket{\psi_s} \in \mathcal{A},
\]
and thus $\ket{\nu_0} \in \mathcal{A}$.

The boundary terms ${\ket{\tils}}\ket{\psi_s}$ and ${\ket{\tilt}}U_s(t)\ket{\psi_s}$ are orthogonal to $\mathcal{A}$. Therefore, the reflection through $\mathcal{A}$ leaves the catalyst $\ket{\nu_0}$ unchanged and puts a minus in front of the boundary terms:
\begin{align*}
    {\ket{\tils}}\ket{\psi_s} + {\ket{\tilt}} U_s(t)\ket{\psi_s} + \ket{\nu_0}
& \overset{2\Pi_{\cal B}-I}{\mapsto} {\ket{\tils}}\ket{\psi_s} + {\ket{\tilt}} U_s(t)\ket{\psi_s} + \ket{\nu_0}\\
&\overset{2\Pi_{\cal A}-I}{\mapsto} -{\ket{\tils}}\ket{\psi_s} - {\ket{\tilt}} U_s(t)\ket{\psi_s} + \ket{\nu_0},
\end{align*}
proving \Cref{eq:transduced_sums1}.

\paragraph{Proof of \Cref{eq:transduced_sums2}.} In order to prove \Cref{eq:transduced_sums2}, we show that
\[ -{\ket{\tils}}\ket{\psi_s} + {\ket{\tilt}}U_s(t)\ket{\psi_s} + \ket{\nu_1} \in \mathcal{A}^{\perp} \cap \mathcal{B}^{\perp}.
\]
This implies that both reflections, through $\mathcal{A}$ and $\mathcal{B}$, put a minus in front of this vector, so $U_{\mathcal{AB}}$ leaves the vector unchanged. First, we check orthogonality to $\mathcal{A}$. For any edge $\{u,v\} \in E$ and any
$z \in [k]$,
\begin{align*}
\left( \ket{u,i_{u,v}}\ket{z} + \ket{v,i_{v,u}}U_{uv}\ket{z} \right)^\dagger \ket{\nu_1} &= \sqrt{w_0}\frac{\theta_{uv}}{\sqrt{w_{uv}}}
\bra{z} U_s(u) \ket{\psi_s} - \sqrt{w_0}\frac{\theta_{uv}}{\sqrt{w_{uv}}} \bra{z}U_{uv}^\dagger U_s(v) \ket{\psi_s}\\
&= \sqrt{w_0}\frac{\theta_{uv}}{\sqrt{w_{uv}}} \bra{z }U_s(u) \ket{\psi_s} - \sqrt{w_0}\frac{\theta_{uv}}{\sqrt{w_{uv}}} \bra{z} U_s(u) \ket{\psi_s}
= 0,
\end{align*}
where we again used $U_s(v)=U_{uv}U_s(u)$. Thus $\ket{\nu_1}$ is orthogonal to ${\cal A}$.
The boundary terms are also orthogonal to $\mathcal{A}$, so the whole vector is orthogonal to $\mathcal{A}$.

Next, we check orthogonality to $\mathcal{B}$. For any vertex $u\in V$, and any $z\in [k]$, we have:
\begin{align*}
   \left( \sum_{i=0}^{d_u - 1} \sqrt{w_{u,f_u(i)}} \ket{u,i}\ket{z} \right)^\dagger
\left( - {\ket{\tils}} \ket{\psi_s} + {\ket{\tilt}} U_s(t) \ket{\psi_s} + \ket{\nu_1} \right)
={}&     \left( \sum_{i=0}^{d_u - 1} \sqrt{w_{u,f_u(i)}} \ket{u,i}\ket{z} \right)^\dagger
\ket{\nu_1}\\
={}& \sqrt{w_0}\sum_{i=0}^{d_u - 1}\theta_{u,f_u(i)}
\bra{z} U_s(u) \ket{\psi_s}\\
={}& \sqrt{w_0}\theta(u)\bra{z}U_s(u)\ket{\psi_s}.
\end{align*}
Thus, for all internal vertices $u\in V\setminus\{s,t\}$, since $\theta$ is an $st$-flow, $\theta(u)=0$, and so
\begin{align*}
   \left( \sum_{i=0}^{d_u - 1} \sqrt{w_{u,f_u(i)}} \ket{u,i}\ket{z} \right)^\dagger
\left( - {\ket{\tils}} \ket{\psi_s} + {\ket{\tilt}} U_s(t) \ket{\psi_s} + \ket{\nu_1} \right)
={}&0.
\end{align*}
For the boundary vertex $s$, since $\theta(s)=1$ (see \Cref{def:flow}) we have
\begin{align*}
&\left( \sqrt{w_0}{\ket{\tils}}\ket{z} + \sum_{i=0}^{d_s - 1} \sqrt{w_{s,f_s(i)}}\ket{s,i}\ket{z} \right)^\dagger \left( - {\ket{\tils}} \ket{\psi_s} + {\ket{\tilt}}U_s(t)\ket{\psi_s} + \ket{\nu_1} \right)\\
={}& -\sqrt{w_0}\braket{z}{\psi_s}+\sqrt{w_0}\bra{z}U_s(s)\ket{\psi_s}=0,
\end{align*}
since $U_s(s)=I$.
For the boundary vertex $t$, since $\theta(t)=-1$, we similarly get
\begin{align*}
&\left( \sqrt{w_0}{\ket{\tilt}}\ket{z} + \sum_{i=0}^{d_t - 1} \sqrt{w_{t,f_t(i)}}\ket{t,i}\ket{z} \right)^\dagger \left(
-{\ket{\tils}}\ket{\psi_s} +
{\ket{\tilt}}U_s(t)\ket{\psi_s} + \ket{\nu_1} \right)\\
={}& \sqrt{w_0}\bra{z} U_s(t)\ket{\psi_s} + \sqrt{w_0}(-1)\bra{z}U_s(t)\ket{\psi_s}
=0.
\end{align*}
Thus, the vector is in $\mathcal{A}^{\perp}\cap\mathcal{B}^{\perp}$, and hence it is fixed by $U_{\mathcal{AB}}$.

\paragraph{Catalyst size.} It remains to bound the size of the catalysts $\ket{\nu_\pm}=\frac{1}{2} \left( \ket{\nu_0} \pm \ket{\nu_1} \right)$. Since all edge basis states are orthogonal,
\[
\norm{\ket{\nu_0}}^2 = \frac{2}{w_0} \sum_{\{u,v\}\in E} w_{uv} = \frac{W(G)}{w_0}.
\]
Similarly, since $\theta$ is an optimal unit $st$-flow,
\[
\norm{\ket{\nu_1}}^2 = 2w_0 \sum_{\{u,v\}\in E} \frac{\theta_{uv}^2}{w_{uv}} = 2w_0 {\cal R}_{st}(G).
\]
Moreover, $\ket{\nu_0}$ and $\ket{\nu_1}$ are orthogonal, since
\[
\braket{\nu_0}{\nu_1} = \sum_{u\in V}\sum_{i=0}^{d_u-1}\theta_{u,f_u(i)}\norm{U_s(u)\ket{\psi_s}}^2 = \sum_{u\in V}\theta(u) = \theta(s)+\theta(t) = 0.
\]
Therefore,
\begin{equation*}
\norm{\frac{1}{2} \left( \ket{\nu_0} \pm \ket{\nu_1} \right)}^2 = \frac{1}{4} \left( \norm{\ket{\nu_0}}^2 + \norm{\ket{\nu_1}}^2 \right) = \frac{1}{4}\left(\frac{W(G)}{w_0} + 2 w_0{\cal R}_{st}(G)\right).
\end{equation*}
\end{proof}
\begin{lemma}\label{lem:gauge-state-preparation-transducer-no-path}
If $s$ and $t$ are not connected, then for each $b\in\{s,t\}$ and any unit vector $\ket{\psi}\in \mathbb{C}^k$, $U_{\cal AB}$ transduces $\ket{\hat b}\ket{\psi}\rightsquigarrow-\ket{\hat b}\ket{\psi}$, with a catalyst $\ket{\nu}\in{\cal L}$ satisfying $\norm{\ket{\nu}}^2\leq\frac{W(G)}{w_0}$, where $\ket{\hat b}:=\ket{b}\ket{\bot}$.
\end{lemma}
\begin{proof}
We give the proof for $b=s$; the case $b=t$ is identical with the roles of $s$ and $t$ exchanged (note that the definitions of ${\cal A}$ and ${\cal B}$ are symmetric in $s$ and $t$). Write $\ket{\psi_s}:=\ket{\psi}$.
Let $V_s$ be the set of vertices connected to $s$, and let $E_s$ be the set of edges with both endpoints in $V_s$. Define
\begin{align*}
    \ket{\nu} &:= \frac{1}{\sqrt{w_0}}\sum_{\{u,v\}\in E_s}\sqrt{w_{uv}}\left( \ket{u,i_{u,v}}U_s(u)\ket{\psi_s} + \ket{v,i_{v,u}}U_s(v)\ket{\psi_s} \right)\\
&= \sum_{\substack{u\in V_s,\\ i\in\{0,\dots,d_u-1\}}}\sqrt{\frac{w_{u,f_u(i)}}{w_0}}\ket{u,i}U_s(u)\ket{\psi_s},
\end{align*}
which is similar to $\ket{\nu_0}$ defined in the proof of \Cref{lem:gauge-state-preparation-transducer}.
First, we show that ${\ket{\tils}}\ket{\psi_s} + \ket{\nu} \in \mathcal{B}$. Grouping the terms in $\ket{\nu}$ by vertices gives:
\begin{align*}
&\sqrt{w_0}({\ket{\tils}}\ket{\psi_s} + \ket{\nu})\\
={}& \left(\sqrt{w_0} {\ket{\tils}}\ket{\psi_s} + \sum_{i = 0}^{d_s - 1} \sqrt{w_{s,f_s(i)}} \ket{s,i}U_s(s)\ket{\psi_s} \right)
+ \sum_{u\in V_s\setminus\{s\}} \left( \sum_{i = 0}^{d_u - 1} \sqrt{w_{u,f_u(i)}}\ket{u,i}U_s(u)\ket{\psi_s} \right)\\
={}& \left(\sqrt{w_0} {\ket{\tils}} + \sum_{i = 0}^{d_s - 1} \sqrt{w_{s,f_s(i)}} \ket{s,i}\right)\ket{\psi_s}
+ \sum_{u\in V_s\setminus\{s\}} \left( \sum_{i = 0}^{d_u - 1} \sqrt{w_{u,f_u(i)}}\ket{u,i} \right)U_s(u)\ket{\psi_s},
\end{align*}
since $U_s(s)=I$. This vector is in $\mathcal{B}$, since $V_s\setminus\{s\}$ contains neither $s$ nor $t$, so the reflection through $\mathcal{B}$ leaves ${\ket{\tils}}\ket{\psi_s}+\ket{\nu}$ unchanged.

Next, observe that $\ket{\nu}\in\mathcal{A}$. Indeed, since the graph is flat, for every edge $\{u,v\} \in E_s$ we have $U_s(v) = U_{uv} U_s(u)$. For every edge $\{u,v\}\in E_s$,
\[
\ket{u,i_{u,v}}U_s(u)\ket{\psi_s} + \ket{v,i_{v,u}}U_s(v)\ket{\psi_s} = \ket{u,i_{u,v}}U_s(u)\ket{\psi_s} + \ket{v,i_{v,u}}U_{uv}U_s(u)\ket{\psi_s} \in \mathcal{A}.
\]
The boundary state ${\ket{\tils}}\ket{\psi_s}$ is orthogonal to $\mathcal{A}$. Therefore, the reflection through $\mathcal{A}$ leaves $\ket{\nu}$ unchanged and puts a minus sign in front of ${\ket{\tils}}\ket{\psi_s}$. Hence
\[
U_{\mathcal{AB}} \left( {\ket{\tils}}\ket{\psi_s} + \ket{\nu}\right) =  - {\ket{\tils}}\ket{\psi_s}+\ket{\nu}.
\]
It remains to bound the catalyst size. Since all edge basis states appearing in $\ket{\nu}$ are orthogonal,
\begin{equation*}
     \norm{\ket{\nu}}^2 = \frac{2}{w_0} \sum_{\{u,v\} \in E_s} w_{uv} \le \frac{W(G)}{w_0}.
\end{equation*}
\end{proof}

\subsection{Implementation of the reflections}\label{sec:implementation}
In this subsection we show that the transducer $U_{\mathcal{AB}}$ can be implemented with a constant number of oracle calls and a {logarithmic} number of elementary gates.
\begin{lemma}\label{lem:gauge-transducer-implementation}
The unitary $U_{\mathcal{AB}} = ( 2\Pi_{\mathcal{A}} - I ) ( 2\Pi_{\mathcal{B}} - I )$ can be implemented using $O(\log n)$ elementary gates and $O(1)$ oracle calls to $O_G$ and $O_w$.
\end{lemma}
\begin{proof}
    The statement of the lemma follows from \lem{Ref_A_transducer} and \lem{Ref_B_transducer} below.
\end{proof}
\begin{lemma}\label{lem:Ref_A_transducer}
    The reflection $2\Pi_{\mathcal{A}} - I$ can be implemented using $O(1)$ elementary gates and one oracle call to $O_G$.
\end{lemma}
\begin{proof}
Recall that
\[
\mathcal{A} = \bigoplus_{\{u,v\} \in E} \mathrm{span} \left\{ \ket{u,i_{u,v}}\ket{z} + \ket{v,i_{v,u}} U_{uv}\ket{z} : z \in [k] \right\}.
\]
On the private space $\mathcal{L}$, the oracle $O_G$ itself is the reflection through $\mathcal{A}$.
Indeed, for any $\{u,v\} \in E$ and $z \in [k]$, define
    \[
        \ket{\phi_{uv}^+(z)}
        = \ket{u,i_{u,v}}\ket{z} + \ket{v,i_{v,u}}U_{uv}\ket{z},
        \qquad
        \ket{\phi_{uv}^-(z)}
        = \ket{u,i_{u,v}}\ket{z} - \ket{v,i_{v,u}}U_{uv}\ket{z}.
    \]
A direct calculation shows
    \begin{samepage}
    \begin{align*}
        O_G \ket{\phi_{uv}^+(z)}
        &= \ket{v,i_{v,u}}U_{uv}\ket{z} + \ket{u,i_{u,v}}U_{vu}U_{uv}\ket{z}
         = \ket{u,i_{u,v}}\ket{z} + \ket{v,i_{v,u}}U_{uv}\ket{z}\\
         &= \ket{\phi_{uv}^+(z)},\\
        O_G \ket{\phi_{uv}^-(z)}
        &= \ket{v,i_{v,u}}U_{uv}\ket{z} - \ket{u,i_{u,v}}U_{vu}U_{uv}\ket{z}
         = -\bigl(\ket{u,i_{u,v}}\ket{z} - \ket{v,i_{v,u}}U_{uv}\ket{z}\bigr)\\
         &= -\ket{\phi_{uv}^-(z)}.
    \end{align*}
    \end{samepage}
    \noindent Thus $$\mathrm{span} \left\{ \ket{u,i_{u,v}}\ket{z} + \ket{v,i_{v,u}} U_{uv}\ket{z} : z \in [k] \right\} = \mathrm{span} \{\ket{\phi_{uv}^+ (z)} : z \in [k]\}$$ is the $(+1)$-eigenspace of $O_G$ on $$\mathrm{span} \left\{ \ket{u,i_{u,v}}\ket{z},\ket{v,i_{v,u}}\ket{z} : z \in [k] \right\},$$ and its orthogonal complement is the $(-1)$-eigenspace. Indeed,
    \begin{align*}
        &\mathrm{span} \{\ket{u,i_{u,v}}\ket{z}, \ket{v,i_{v,u}}\ket{z} : z\in[k]\}\\
        =& \mathrm{span} \{\ket{u,i_{u,v}}\ket{z} : z\in[k] \} \oplus \mathrm{span} \{ \ket{v,i_{v,u}}\ket{z} : z\in[k]\}\\
        =& \mathrm{span} \{\ket{u,i_{u,v}}\ket{z} : z\in[k] \} \oplus \mathrm{span} \{ \ket{v,i_{v,u}} U_{uv} \ket{z} : z\in[k]\}\\
        =& \mathrm{span} \{\ket{u,i_{u,v}}\ket{z}, \ket{v,i_{v,u}} U_{uv} \ket{z} : z\in[k]\}\\
        =& \mathrm{span} \{\ket{u,i_{u,v}}\ket{z} + \ket{v,i_{v,u}} U_{uv} \ket{z}, \ket{u,i_{u,v}}\ket{z} - \ket{v,i_{v,u}} U_{uv} \ket{z} : z\in[k]\}\\
        =& \mathrm{span} \{\ket{\phi_{uv}^+ (z)} : z \in [k]\} \oplus \mathrm{span} \{\ket{\phi_{uv}^- (z)} : z \in [k]\}.
    \end{align*}
    Since the subspaces $\mathrm{span} \{\ket{u,i_{u,v}}\ket{z}, \ket{v,i_{v,u}}\ket{z} : z\in[k]\}$ are pairwise orthogonal, the global $(+1)$-eigenspace of $O_G$ in $\bigoplus_{\{u,v\} \in E} \mathrm{span} \{\ket{u,i_{u,v}}\ket{z}, \ket{v,i_{v,u}}\ket{z} : z\in[k]\}$ is precisely
    \[
        \bigoplus_{\{u,v\} \in E} \mathrm{span} \left\{ \ket{u,i_{u,v}}\ket{z} + \ket{v,i_{v,u}} U_{uv}\ket{z} : z \in [k] \right\} = \mathcal{A},
    \]
    and the $(-1)$-eigenspace is $\mathcal{A}^\perp$ within $\mathcal{L}$.
The only remaining issue is that the public space
\[
\mathcal{H} = \mathrm{span}\{ \ket{\tils}\ket{z}, \ket{\tilt}\ket{z} : z \in [k] \}
\]
is orthogonal to $\mathcal{A}$, and hence must receive a minus sign under $2\Pi_{\mathcal{A}} - I$. Since $\bot$ is not a valid neighbour index, $O_G$ acts as the identity on $\mathcal{H}$, so
\[
2\Pi_{\mathcal{A}}-I = O_G\,(I-2\Pi_{\mathcal{H}}),\qquad \Pi_{\mathcal{H}}:=\big(\ket{\tils}\bra{\tils}+\ket{\tilt}\bra{\tilt}\big)\otimes I_{\mathbb{C}^k},
\]
on $\mathcal{H}\oplus\mathcal{L}$. The reflection $I-2\Pi_{\mathcal{H}}$ is a phase flip on the states whose first two registers are $\ket{\tils}$ or $\ket{\tilt}$, regardless of the last register, which by our assumption on $\ket{\tils},\ket{\tilt}$ takes $O(1)$ elementary gates. Hence $2\Pi_{\mathcal{A}}-I$ can be implemented in $O(1)$ elementary gates and one call to $O_G$.
\end{proof}
\begin{lemma}\label{lem:Ref_B_transducer}
    The reflection $2\Pi_{\mathcal{B}} - I$ can be implemented using $O(\log n)$ elementary gates and two oracle calls to $O_w$.
\end{lemma}
\begin{proof}
We show how to prepare a working basis for $\mathcal{B}$. For $u \in V \setminus \{s,t\}$ and $z \in [k]$, the corresponding normalized basis vector is
\[
\frac{1}{\sqrt{w_u}} \sum_{i=0}^{d_u - 1} \sqrt{w_{u,f_u(i)}}\ket{u,i}\ket{z}.
\]
This can be prepared from $\ket{u}\ket{0}\ket{z}$ using one call to $O_w$. If $u \notin \{s,t\}$ is isolated, we take $O_w$ to act as the identity on $\ket{u}\ket{0}$. The additional vector $\ket{u}\ket{0}\ket{z}$ is orthogonal to ${\cal H}\oplus{\cal L}$, since $u\notin\{s,t\}$ and $u$ has no incident edges, whereas ${\cal H}$ is supported on boundary edges and ${\cal L}$ on ordinary graph edges, so it does not affect the implemented reflection on that space.

For $b \in \{s,t\}$, if $w_b = 0$, the rotation prepares $\ket{b}\ket{\bot}\ket{z}$ directly. For $b \in \{s,t\}$, such that $w_b \neq 0$, and $z \in [k]$, the corresponding normalized basis vector is
\[
\frac{1}{\sqrt{w_0 + w_b}} \left( \sqrt{w_0} \ket{b}\ket{\bot}\ket{z} +  \sum_{i=0}^{d_b - 1} \sqrt{w_{b,f_b(i)}}\ket{b,i}\ket{z} \right).
\]
This can be prepared from $\ket{b}\ket{0}\ket{z}$ by first using a rotation, controlled on the first register being $s$ or $t$, to create the superposition $\frac{\sqrt{w_0}}{\sqrt{w_0 + w_b}}\ket{b}\ket{\bot}\ket{z}+\frac{\sqrt{w_b}}{\sqrt{w_0 + w_b}}\ket{b}\ket{0}\ket{z}$ (for instance, if $\bot$ is encoded by a flag qubit, this is a single-qubit rotation on the flag), and then applying $O_w$, controlled on the second register not being $\bot$, to prepare
\[
\frac{1}{\sqrt{w_b}} \sum_{i=0}^{d_b - 1} \sqrt{w_{b,f_b(i)}}\ket{b,i}
\]
on the second branch.
Since $w_s$, $w_t$, and $w_0=\sqrt{\mathbf W/\mathbf R}$ are known classical quantities, the two required rotation angles, corresponding to $b=s$ and $b=t$, can be computed in advance and hard-coded into the circuit. Thus, under our assumption that $s$ and $t$ can be recognized using $O(1)$ elementary gates, the preparation uses two controlled single-qubit rotations, $O(1)$ additional elementary gates, and one call to $O_w$. Hence a single circuit $P$, using one call to $O_w$ and $O(1)$ other gates, maps $\ket{u}\ket{0}\ket{z}$ to the corresponding basis vector of $\mathcal{B}$ for every non-isolated vertex $u$ and each boundary vertex $u\in\{s,t\}$, and every $z\in[k]$. On $\mathcal{H}\oplus\mathcal{L}$, we therefore have
\[
2 \Pi_{\mathcal{B}} - I = P (2 \Pi_0 - I) P^\dagger,
\]
where $\Pi_0$ is the projector onto $\mathrm{span} \{\ket{u}\ket{0}\ket{z} : u \in V,z \in [k]\}$.
The reflection $2\Pi_0-I$ is the operation that applies phase $+1$ when the second register is $\ket{0}$ and phase $-1$ otherwise, and thus it can be implemented by a phase flip controlled on the $O(\log n)$ qubits of the second register, using $O(\log n)$ elementary gates. Therefore, $2\Pi_{\mathcal{B}} - I$ can be implemented using $O(\log n)$ elementary gates and two calls to $O_w$ (one each for $P$ and $P^\dagger$).
\end{proof}
We now combine the results established throughout this section to complete the proof of \thm{gauge-state-preparation-transducer}.
\begin{proof}[Proof of \thm{gauge-state-preparation-transducer}]
Item~1 is by definition of ${\cal H}$ and ${\cal L}$ in \Cref{sec:transducer}, and item~3 is \Cref{lem:gauge-transducer-implementation}. For item~2, fix $w_0=\sqrt{\mathbf{W}/\mathbf{R}}$, so that, by the remark after \Cref{lem:gauge-state-preparation-transducer}, the catalysts in \Cref{lem:gauge-state-preparation-transducer,lem:gauge-state-preparation-transducer-no-path} have squared norm at most $\frac34\sqrt{\mathbf{WR}}$ and $\sqrt{\mathbf{WR}}$, respectively.

Suppose $s$ and $t$ are connected. As $\ket{\psi_s}$ ranges over unit vectors in $\mathbb{C}^k$, so does $U_s(t)\ket{\psi_s}$, so \Cref{lem:gauge-state-preparation-transducer} says that $U_{\cal AB}$ transduces $\ket{\tils}\ket{\psi}\rightsquigarrow-\ket{\tilt}U_s(t)\ket{\psi}$ and $\ket{\tilt}\ket{\phi}\rightsquigarrow-\ket{\tils}U_s(t)^\dagger\ket{\phi}$ for all unit $\ket{\psi},\ket{\phi}\in\mathbb{C}^k$. By \Cref{thm:transduction-action}, the transduction action $\Sigma$ is linear, so this determines $\Sigma$ on a spanning set of ${\cal H}$, and it is the stated operator. If $s$ and $t$ are not connected, \Cref{lem:gauge-state-preparation-transducer-no-path} gives $\Sigma\ket{\hat b}\ket{\psi}=-\ket{\hat b}\ket{\psi}$ for $b\in\{s,t\}$ and all $\ket{\psi}$, so $\Sigma=-I_{\cal H}$. In both cases $\Sigma$ is a Hermitian involution, since $\braket{\tils}{\tilt}=0$.

Finally, we bound $W(U_{\cal AB})$. Any unit vector in ${\cal H}$ can be written as $\ket{\xi}=a\ket{\tils}\ket{\psi}+b\ket{\tilt}\ket{\phi}$ with $\ket{\psi},\ket{\phi}$ unit vectors and $|a|^2+|b|^2=1$. Since the smallest catalyst depends linearly on the state (\Cref{sec:transducers}), $\ket{\xi}$ has a catalyst $a\ket{\nu(\tils,\psi)}+b\ket{\nu(\tilt,\phi)}$, whose squared norm is at most $(|a|+|b|)^2\sqrt{\mathbf{WR}}\leq 2\sqrt{\mathbf{WR}}$ by the triangle inequality. Hence $W(U_{\cal AB})\leq 2\sqrt{\mathbf{WR}}$.

\end{proof}

\subsection{Composition with Amplitude Estimation}
\label{sec:proof-main-theorem}

In this section, we prove \Cref{cor:gauge-state-overlap-algorithm}, which we restate below for convenience.

\bigskip

\noindent\textbf{\Cref*{cor:gauge-state-overlap-algorithm} (restated).}
{\it There is a quantum algorithm that solves \Cref{prob:gauge_problem_estimate}
with bounded error, meaning that, with probability at least $2/3$, it correctly
reports that $s$ and $t$ are disconnected, or outputs an estimate $\tilde p$ such that
$\abs{\,\abs{\bra{\psi_t}U_s(t)\ket{\psi_s}}^2-\tilde p\,}\leq \eps$, using
$\widetilde{O}\left(\sqrt{\mathbf{W}\mathbf{R}}/\eps\right)$ calls to $O_G$,
$O_w$, $O_s$ and $O_t$, and
$\widetilde{O}\left(\sqrt{\mathbf{W}\mathbf{R}}/\eps\right)$ additional
elementary operations.}

\bigskip

\noindent Throughout, let $\Sigma$ be the transduction action of $U_{\cal AB}$ on ${\cal H}$, as in
\Cref{thm:gauge-state-preparation-transducer}. The algorithm will consist of two stages. First, use  the transducer from \Cref{thm:gauge-state-preparation-transducer} to get a subroutine to check if $s$ and $t$ are connected. If they are, compose the transducer $U_{\cal AB}$, which we then know generates $U_s(t)\ket{\psi_s}$, with amplitude estimation, via \Cref{lem:amp-amp-comp}, to get a transducer -- and from it, an algorithm -- for estimating $|\bra{\psi_t}U_s(t)\ket{\psi_s}|^2$.

\paragraph{Detecting disconnectedness.} Apply
\Cref{thm:transducer-implementation} to the transducer $S=U_{\cal AB}$ from \Cref{thm:gauge-state-preparation-transducer} with error parameter
$\eps_1=10^{-4}$ and $K_1=O\big(\sqrt{\mathbf{W}\mathbf{R}}\big)$ sufficiently large, on the
initial state ${\ket{\tils}}\,O_s\ket{0}={\ket{\tils}}\ket{\psi_s}$, and measure
the first register. We will require $K_1\geq \frac{4}{\eps_1} W(S,\ket{\hat{s}}O_s\ket{0})$ for the correctness guarantee of \Cref{thm:transducer-implementation} to hold, and since $W(S)\leq 2\sqrt{\mathbf{WR}}$, this can be satisfied by some $K_1=O(\sqrt{\mathbf{WR}})$. Then
by \Cref{thm:gauge-state-preparation-transducer}, the resulting state is
within distance ${\sqrt{\eps_1}}=10^{-2}$ of $-{\ket{\tilt}}\,U_s(t)\ket{\psi_s}$ if $s$ and $t$
are connected, and of $-{\ket{\tils}}\ket{\psi_s}$ otherwise, so the measurement
returns $t$ in the former case, and $s$ in the latter, with probability at
least $(1-\sqrt{\eps_1})^2\geq 0.98$. If the outcome is $s$, output ``disconnected'' and stop;
otherwise, continue. By \Cref{lem:gauge-transducer-implementation}, this step
costs $\widetilde{O}\big(\sqrt{\mathbf{W}\mathbf{R}}\big)$ oracle calls and other gates.

\paragraph{Amplitude estimation setup.} Now the idea will be to apply amplitude estimation, via the transducer in \Cref{lem:amp-amp-comp}. Let
\begin{equation}\label{eq:pi-Pi-R}
\ket{\pi}:=(I\otimes O_t^\dagger)\Sigma (I\otimes O_s)\ket{\tils}\ket{0}=-\ket{\tilt}O_t^\dagger U_s(t)\ket{\psi_s}
\quad \mbox{and}\quad
\Pi := \ket{\tilt}\bra{\tilt}\otimes \ket{0}\bra{0}.
\end{equation}
To apply \Cref{lem:amp-amp-comp}, we need a transducer $S_O$ with transduction action
\begin{equation}\label{eq:ref-pi}
    2\ket{\pi}\bra{\pi}-I = (I\otimes O_t^\dagger)\Sigma \underbrace{ (I\otimes O_s)(2\ket{\tils}\bra{\tils}\otimes \ket{0}\bra{0}-I)(I\otimes O_s^\dagger)}_{=:R}\Sigma^\dagger (I\otimes O_t).
\end{equation}
Then amplitude estimation with this setup will estimate:
\[
p:=\norm{\Pi \ket{\pi}}^2
=\abs{\bra{0}O_t^\dagger U_s(t)\ket{\psi_s}}^2
=\abs{\bra{\psi_t}U_s(t)\ket{\psi_s}}^2.
\]

\begin{claim}\label{claim:above}
    Let $S_O$ be a transducer on ${\cal H}\oplus {\cal L}_0\oplus {\cal L}_1$, where ${\cal L}_b$ is isomorphic to ${\cal L}$, with public space ${\cal H}$, defined
    \begin{align*}
        S_O&=c_{\cal H}(I\otimes O_t^\dagger)\cdot c_{{\cal H}\oplus {\cal L}_1}U_{\cal AB} \cdot c_{\cal H} R
        \cdot c_{{\cal H}\oplus {\cal L}_0}U_{\cal AB}^\dagger
        \cdot c_{\cal H}(I\otimes O_t)
    \end{align*}
    where $c_{\cal S}$ indicates a control on being in ${\cal S}$, so $c_{\cal S}U=U_{\cal S}\oplus I_{{\cal S}^\bot}$, and $R$ is as defined in \Cref{eq:ref-pi}.
    Then $S_O$ has transduction action $2\ket{\pi}\bra{\pi}-I$, and $W(S_O)\leq 2W(U_{\cal AB})$.
    Moreover, $S_O$ can be implemented using one controlled call to each of $U_{\cal AB}$, $U_{\cal AB}^\dagger$, $O_s$, $O_s^\dagger$, $O_t$ and $O_t^\dagger$, and $O(\log k)$ additional elementary gates.
\end{claim}
Operationally, to implement these controls, we can imagine having a 3-dimensional register that indicates which of ${\cal H}$, ${\cal L}_0$, or ${\cal L}_1$ a state is in. Then if we apply, for example, $U_{\cal AB}$ to a state $\ket{0}\ket{\xi}+\ket{2}\ket{v}\in {\cal H}\oplus {\cal L}_1$, it should act as $U_{\cal AB}(\ket{\xi}+\ket{v})$. This is possible to implement, since ${\cal H}$ and ${\cal L}$ are orthogonal, and it is simple to compute in which one you are, in order to temporarily uncompute the trit: a state is in ${\cal H}$ if and only if its index register holds $\bot$, which can be checked in $O(1)$ gates. In the proof below, we leave this implicit, and use subscripts to indicate to which space a state belongs.
\begin{proof}
    A similar statement follows from \cite[Proposition~9.1]{belovs2023LasVegasTime}, which says that
for a transducer $U_{\cal AB}$ with transduction action $\Sigma$, $U_{\cal AB}^\dagger$ is a transducer with action $\Sigma^\dagger$ and the same transduction complexity; combined with \cite[Proposition~9.9]{belovs2023LasVegasTime}, for sequential composition. Both are stated for canonical transducers, and as the argument is rather simple, we give an explicit proof, which does not need the canonical form.

Fix any $\ket{\xi}\in {\cal H}$. We know there exists a catalyst $\ket{v}\in {\cal L}$ for the state ${\Sigma^\dagger}(I\otimes O_t)\ket{\xi}\in {\cal H}$, such that
\begin{align*}
    U_{\cal AB}(\Sigma^\dagger (I\otimes O_t)\ket{\xi}+\ket{v})&=(I\otimes O_t)\ket{\xi}+\ket{v}\\
\Sigma^\dagger (I\otimes O_t)\ket{\xi}+\ket{v}&=U_{\cal AB}^\dagger ((I\otimes O_t)\ket{\xi}+\ket{v})
\end{align*}
and $\norm{\ket{v}}^2\leq W(U_{\cal AB})$. Similarly, 
there exists a catalyst $\ket{v'}\in {\cal L}$ such that
$$U_{\cal AB}(R\Sigma^\dagger (I\otimes O_t)\ket{\xi}+\ket{v'})=
\Sigma R\Sigma^\dagger (I\otimes O_t)\ket{\xi}+\ket{v'}$$
and $\norm{\ket{v'}}^2\leq W(U_{\cal AB})$.

Thus,  we have:
\begin{align*}
    S_O:\ket{\xi}_{\cal H}+\ket{v}_{{\cal L}_0}+\ket{v'}_{{\cal L}_1} \overset{c_{\cal H}(I\otimes O_t)}{\longmapsto}{}& (I\otimes O_t)\ket{\xi}_{\cal H}+\ket{v}_{{\cal L}_0}+\ket{v'}_{{\cal L}_1}\\
    \overset{c_{{\cal H}\oplus {\cal L}_0}U_{\cal AB}^\dagger}{\longmapsto}{}& U_{\cal AB}^\dagger((I\otimes O_t)\ket{\xi}_{\cal H}+\ket{v}_{{\cal L}_0})+\ket{v'}_{{\cal L}_1}\\
    ={}& \Sigma^\dagger(I\otimes O_t)\ket{\xi}_{\cal H}+\ket{v}_{{\cal L}_0}+\ket{v'}_{{\cal L}_1}\\
    \overset{c_{{\cal H}}R}{\longmapsto}{}& R\Sigma^\dagger(I\otimes O_t)\ket{\xi}_{\cal H}+\ket{v}_{{\cal L}_0}+\ket{v'}_{{\cal L}_1}\\
    \overset{c_{{\cal H}\oplus {\cal L}_1}U_{\cal AB}}{\longmapsto}{}& U_{\cal AB}(R\Sigma^\dagger(I\otimes O_t)\ket{\xi}_{\cal H}+\ket{v'}_{{\cal L}_1})+\ket{v}_{{\cal L}_0}\\
    ={}& \Sigma R\Sigma^\dagger(I\otimes O_t)\ket{\xi}_{\cal H}+\ket{v'}_{{\cal L}_1}+\ket{v}_{{\cal L}_0}\\
    \overset{c_{{\cal H}}(I\otimes O_t^\dagger)}{\longmapsto}{}&(I\otimes O_t^\dagger)\Sigma R\Sigma^\dagger(I\otimes O_t)\ket{\xi}_{\cal H}+\ket{v}_{{\cal L}_0}+\ket{v'}_{{\cal L}_1}\\
    ={}& (2\ket{\pi}\bra{\pi}-I)\ket{\xi}_{\cal H}+\ket{v}_{{\cal L}_0}+\ket{v'}_{{\cal L}_1}.
\end{align*}
Thus, $S_O$ has transduction action $2\ket{\pi}\bra{\pi}-I$, and since $\norm{\ket{v}_{{\cal L}_0}+\ket{v'}_{{\cal L}_1}}^2 = \norm{\ket{v}}^2+\norm{\ket{v'}}^2\leq 2W(U_{\cal AB})$,
$W(S_O)\leq 2W(U_{\cal AB})$, as claimed.

For the implementation cost, $S_O$ is a product of five controlled operations. The controls cost $O(1)$ gates each, as discussed above. The operation $R$ is the reflection around $\ket{\tils}\ket{0}$, conjugated by $I\otimes O_s$, and so costs one call to each of $O_s$ and $O_s^\dagger$ and $O(\log k)$ additional gates. The remaining four operations are one call to each of $U_{\cal AB}$, $U_{\cal AB}^\dagger$, $O_t$ and $O_t^\dagger$.
\end{proof}

\paragraph{Composition with amplitude estimation.} It remains to compose the transducer $S_O$ from the above claim into amplitude estimation, which we do via \Cref{lem:amp-amp-comp}. For $M=\lceil 8\pi/\eps \rceil$, let $S_M$ be the transducer from \Cref{lem:amp-amp-comp}, which has public space ${\cal H}_M:=\mathrm{span}\{\ket{0}_{\cal S}\}\otimes {\cal T}\otimes {\cal H}$ and private space ${\cal L}_M:=\big(\mathrm{span}\{\ket{s}_{\cal S}:s\geq1\}\otimes{\cal T}\otimes{\cal H}\big)\oplus\big({\cal S}\otimes {\cal T}\otimes {\cal L}_O\big)$, where ${\cal L}_O:={\cal L}_0\oplus {\cal L}_1$ is the private space of $S_O$. Its transduction action maps $\ket{0}\ket{0}\ket{\pi}$ to $\ket{0}\ket{\eta}$, so we will do one final composition to get a transducer that implements $\ket{0}\ket{0}\ket{\tils,0}\rightsquigarrow \ket{0}\ket{\eta}$.

\begin{table}
\centering
    \begin{tabular}{c|c|c|c}
       Transducer & Definition & Public Space & Private Space\\
       $S=U_{\cal AB}$  & \Cref{thm:gauge-state-preparation-transducer} & ${\cal H}$ & ${\cal L}$\\ 
       $S_O$ & \Cref{claim:above} & ${\cal H}_O={\cal H}$ & ${\cal L}_O={\cal L}\oplus{\cal L}$\\
       $S_M$ & \Cref{lem:amp-amp-comp} & ${\cal H}_M={\cal S}_0\otimes {\cal T}\otimes {\cal H}$ & ${\cal L}_M=\Big({\cal S}_{\geq 1}\otimes {\cal T}\otimes {\cal H}\Big)\oplus\Big({\cal S}\otimes{\cal T}\otimes {\cal L}_O\Big)$\\
       $S_{\rm fin}$ & \Cref{claim:S-fin} & ${\cal H}_{\rm fin}={\cal H}_M$ & ${\cal L}_{\rm fin}=\Big({\cal S}_{\geq 1}\otimes {\cal T}\otimes {\cal H}\Big)\oplus\Big({\cal S}\otimes{\cal T}\otimes {\cal L}_O\oplus {\cal L}\Big)$
    \end{tabular}
    \caption{Summary of public and private spaces for the various transducers in this section, where ${\cal S}_0=\mathrm{span}\{\ket{0}_{\cal S}\}$ and ${\cal S}_{\geq 1}=\mathrm{span}\{\ket{1}_{\cal S},\dots,\ket{M-1}_{\cal S}\}$.}\label{fig:pub-priv-spaces}
\end{table}
\begin{claim}\label{claim:S-fin}
    Let $S_{\rm fin}$ be the transducer on ${\cal S}\otimes {\cal T}\otimes ({\cal H}\oplus {\cal L}_O\oplus {\cal L})$ with public space ${\cal H}_{\rm fin}={\cal H}_M$, defined
    $$S_{\rm fin}=c_{{\cal ST}\otimes ({\cal H}\oplus {\cal L}_O)}S_M 
    \cdot c_{{\cal H}_M}(I_{\cal ST}\otimes (I\otimes O_t^\dagger))
    \cdot c_{{\cal S}_0\otimes{\cal T}\otimes ({\cal H}\oplus{\cal L})}(I_{\cal ST}\otimes U_{\cal AB})
    \cdot c_{{\cal H}_M}(I_{\cal ST}\otimes (I\otimes O_s)),$$
    where ${\cal S}_0=\mathrm{span}\{\ket{0}_{\cal S}\}$.
    Then $S_{\rm fin}$ has transduction action $\ket{0}\ket{0}\ket{{\tils},0}\rightsquigarrow \ket{0}\ket{\eta}$, and $W(S_{\rm fin},\ket{0}\ket{0}\ket{{\tils},0})= O(M\sqrt{\mathbf{WR}})$.
    Moreover, $S_{\rm fin}$ can be implemented using $O(1)$ controlled calls to each of $U_{\cal AB}^{\pm 1}$, $O_s^{\pm 1}$ and $O_t^{\pm 1}$, and $O(\log k + \log^2 M)$ additional elementary gates. In particular, by \Cref{lem:gauge-transducer-implementation}, one call to $S_{\rm fin}$ costs $O(1)$ calls to each of $O_G$, $O_w$, $O_s$ and $O_t$, and $\widetilde{O}(1)$ additional elementary gates.
\end{claim}
\begin{proof}
$S_{\rm fin}$ acts on ${\cal H}_{\rm fin}\oplus {\cal L}_{\rm fin}$, but it will be useful to decompose its private space into two parts, so its total space is (see \Cref{fig:pub-priv-spaces}):
$${\cal H}_{\rm fin}\oplus {\cal L}_{\rm fin}=\underbrace{({\cal S}_0\otimes {\cal T}\otimes {\cal H})}_{{\cal H}_{\rm fin}}\oplus \underbrace{ ({\cal S}\otimes {\cal T}\otimes {\cal L})\oplus {\cal L}_M}_{{\cal L}_{\rm fin}}.$$
The operations defining $S_{\rm fin}$ each act on one or more of these spaces.

    The proof is similar to the previous claim. Let $\ket{v}\in {\cal L}$ be such that
    \begin{equation}\label{eq:S-fin-cat-1}
    U_{\cal AB}(\ket{{\tils},\psi_s}+\ket{v})=\Sigma \ket{\tils,\psi_s}+\ket{v},
    \end{equation}
    with $\norm{\ket{v}}^2\leq W(U_{\cal AB})$.

    Similarly, let $\ket{v'}\in {\cal L}_M=\big({\cal S}_{\geq1}\otimes{\cal T}\otimes{\cal H}\big)\oplus\big({\cal S}\otimes {\cal T}\otimes {\cal L}_O\big)$ be such that
    \begin{equation}\label{eq:S-fin-cat-2}
        S_M(\ket{0,0}\ket{\pi}+\ket{v'})=\ket{0}\ket{\eta}+\ket{v'},
    \end{equation}
    and $\norm{\ket{v'}}^2\leq (M-1)(1+W(S_O))$, as we know exists by \Cref{lem:amp-amp-comp}. By \Cref{claim:above}, $W(S_O)\leq 2W(U_{\cal AB})$, so $\norm{\ket{v'}}^2\leq (M-1)(1+2W(U_{\cal AB}))$.
    Combining these two catalysts into $\ket{v''}=\ket{0,0}_{\cal ST}\ket{v}_{\cal L}+\ket{v'}_{{\cal L}_M}$, which we can easily verify is in ${\cal L}_{\rm fin}$ (see \Cref{fig:pub-priv-spaces}), consider the action of $S_{\rm fin}$ on 
    $$\ket{0,0}_{\cal ST}\ket{\tils,0}_{\cal H} + \ket{v''}_{{\cal L}_{\rm fin}}=\underbrace{\ket{0,0}_{\cal ST}\ket{\tils,0}_{\cal H}}_{\in {\cal H}_{\rm fin}}+ \underbrace{\ket{0,0}_{\cal ST}\ket{v}_{\cal L}+\ket{v'}_{{\cal L}_M}}_{=:\ket{v''}}.$$
    First, we apply $I_{\cal ST}\otimes (I\otimes O_s)$ controlled on being in ${\cal H}_M={\cal H}_{\rm fin}$, which only impacts the first term:
    \begin{align*}
        & \ket{0,0}_{\cal ST}\ket{\tils,0}_{\cal H} + \ket{0,0}_{\cal ST}\ket{v}_{\cal L}+\ket{v'}_{{\cal L}_M}\\
        \overset{c(I\otimes O_s)}{\longmapsto}{}&
        \ket{0,0}_{\cal ST}(I\otimes O_s)\ket{{\tils},0}_{\cal H} + \ket{0,0}_{\cal ST}\ket{v}_{\cal L}+\ket{v'}_{{\cal L}_M}.
    \end{align*}
    Next, we apply $I_{\cal ST}\otimes U_{\cal AB}$, controlled on being in ${\cal S}_0\otimes {\cal T}\otimes ({\cal H}\oplus{\cal L})$, which only acts non-trivially on the first two terms:    
    \begin{align*}
        \overset{cU_{\cal AB}}{\longmapsto}{}& \ket{0,0}_{\cal ST} U_{\cal AB}((I\otimes O_s)\ket{{\tils},0}_{\cal H}+\ket{v}_{\cal L})+\ket{v'}_{{\cal L}_M}\\
        ={}& \ket{0,0}_{\cal ST}\Sigma (I\otimes O_s)\ket{\hat s,0}_{\cal H}+\ket{0,0}_{\cal ST}\ket{v}_{\cal L}+\ket{v'}_{{\cal L}_M}
    \end{align*}
    by \eqref{eq:S-fin-cat-1}, since $O_s\ket{0}=\ket{\psi_s}$. Continuing, we apply $I_{\cal ST}\otimes (I\otimes O_t^\dagger)$ controlled on being in ${\cal H}_M={\cal H}_{\rm fin}$, which again only impacts the first term:     
    \begin{align*}
        \overset{c(I\otimes O_t^\dagger)}{\longmapsto}{}& \ket{0,0}_{\cal ST} (I\otimes O_t^\dagger)\Sigma(I\otimes O_s)\ket{{\tils},0}_{\cal H}+\ket{0,0}_{\cal ST}\ket{v}_{\cal L}+\ket{v'}_{{\cal L}_M}\\
        ={}& \ket{0,0}_{\cal ST} \ket{\pi}_{\cal H}+\ket{0,0}_{\cal ST}\ket{v}_{\cal L}+\ket{v'}_{{\cal L}_M},
    \end{align*}
    by \eqref{eq:pi-Pi-R}. Finally, to complete the application of $S_{\rm fin}$, we apply $S_M$ controlled on being in ${\cal S}_0\otimes{\cal T}\otimes {\cal H}$ or ${\cal L}_M$, which impacts the first and last term:
    \begin{align*}
        \overset{cS_M}{\longmapsto}{}& S_M(\ket{0,0}_{\cal ST} \ket{\pi}_{\cal H}+\ket{v'}_{{\cal L}_M})+\ket{0,0}_{\cal ST}\ket{v}_{\cal L}\\
        ={}& \ket{0}_{\cal S}\ket{\eta}+\ket{v'}_{{\cal L}_M}+\ket{0,0}_{\cal ST}\ket{v}_{\cal L}
    \end{align*}
    by \eqref{eq:S-fin-cat-2}.
    This establishes the claimed transduction action. To compute the transduction complexity, we bound:
    \begin{align*}
        \norm{\ket{v''}}^2=\norm{\ket{v'}_{{\cal L}_M}+\ket{0,0}_{\cal ST}\ket{v}_{\cal L}}^2 &= \norm{\ket{v'}}^2+\norm{\ket{v}}^2
        \leq (M-1)(1+2W(U_{\cal AB}))+W(U_{\cal AB})\\
        &= O(M W(U_{\cal AB}))=O(M\sqrt{\mathbf{WR}}),
    \end{align*}
    by \Cref{thm:gauge-state-preparation-transducer}.

    For the implementation cost, the three operations to the right of $S_M$ are one controlled call to each of $O_s$, $U_{\cal AB}$ and $O_t^\dagger$, with $O(1)$-gate controls as in \Cref{claim:above}, plus the control on $\ket{0}_{\cal S}$, which costs $O(\log M)$ gates. By \Cref{lem:amp-amp-comp}, $S_M$ costs one controlled call to $S_O$, one controlled call to $I-2\Pi$, two Fourier transforms and $O(\log M)$ other gates. Here $I-2\Pi$ is the reflection around $\ket{\tilt}\ket{0}$, costing $O(\log k)$ gates, the Fourier transforms cost $O(\log^2 M)$ gates, and by \Cref{claim:above}, $S_O$ costs $O(1)$ calls to $U_{\cal AB}^{\pm1}$, $O_s^{\pm1}$, $O_t^{\pm1}$, and $O(\log k)$ other gates. Summing gives the claimed cost.
\end{proof}

\paragraph{Transducer to algorithm.} We now turn the transducer $S_{\rm fin}$ into an algorithm, via \Cref{thm:transducer-implementation}, applied with the constant error parameter $\eps_2=10^{-4}$, on the initial state $\ket{0}_{\cal S}\ket{0}_{\cal T}\ket{\tils,0}$, which can be prepared with $O(1)$ gates. By \Cref{claim:S-fin}, it suffices to take some sufficiently large
$$K = O\left(\frac{M\sqrt{\mathbf{WR}}}{\eps_2}\right)=O\left(\frac{\sqrt{\mathbf{WR}}}{\eps}\right)$$
(recall $M=\lceil 8\pi/\eps\rceil$) calls to $S_{\rm fin}$, after which the resulting state $\ket{\eta'}$ satisfies $\norm{\ket{\eta'}-\ket{0}_{\cal S}\ket{\eta}}^2\leq 10^{-4}$. We measure the ${\cal T}$ register of $\ket{\eta'}$ to obtain $z\in\{0,\dots,M-1\}$, and output $\tilde p:=\sin^2(\pi z/M)$.

\paragraph{Correctness.} By \Cref{lem:amp-amp-comp}, if we measured the ${\cal T}$ register of the ideal state $\ket{0}_{\cal S}\ket{\eta}$, then with probability at least $3/4$ we would obtain $z$ such that, 
$$|\tilde p -p|\leq \frac{6\pi}{M}+\frac{9\pi^2}{M^2}
\leq\frac{6\pi}{8\pi}\eps+\frac{9\pi^2}{64\pi^2}\eps^2
\leq \eps.$$
Since $\ket{\eta'}$ is within distance $10^{-2}$ of $\ket{0}_{\cal S}\ket{\eta}$, and detecting connectedness succeeds with probability very close to 1, it is easy to verify that the above still has probability at least $2/3$. 

\paragraph{Cost.} The disconnectedness test costs $\widetilde{O}(\sqrt{\mathbf{WR}})$ calls to $O_G$, $O_w$ and $O_s$, and $\widetilde{O}(\sqrt{\mathbf{WR}})$ other gates. The estimation step makes $K=O(\sqrt{\mathbf{WR}}/\eps)$ controlled calls to $S_{\rm fin}$, and uses $O(K)$ other gates. By \Cref{claim:S-fin}, each call to $S_{\rm fin}$ costs $O(1)$ calls to each of $O_G$, $O_w$, $O_s$ and $O_t$, and $O(\log n+\log k+\log^2 M)=\widetilde{O}(1)$ other gates. The total is therefore
$\widetilde{O}\big(\sqrt{\mathbf{W}\mathbf{R}}/\eps\big)$ calls to each oracle
and $\widetilde{O}\big(\sqrt{\mathbf{W}\mathbf{R}}/\eps\big)$ additional
elementary operations. This completes the proof of \Cref{cor:gauge-state-overlap-algorithm}.

\section{Metropolis-Hastings connection graphs}\label{sec:MH-connection}

The algorithm of \Cref{sec:gauge-state-preparation-transducer} solves $st$-transport in $\widetilde O(\sqrt{\mathbf{W}\mathbf{R}}/\eps)$ time, so its cost depends on the edge weights. If we simply give every edge of $G$ weight $1$, then $W(G)=2|E|$, which can be as large as $n^2$, and ${\cal R}_{st}(G)$ is the usual effective resistance, which can be as large as $n$. Both can happen at once, for instance when $s$ lies in a clique on $n/2$ vertices and $t$ is at the end of a path of length $n/2$ attached to it, and then we only get $\widetilde O(n^{3/2}/\eps)$.

However, the answer to $st$-transport does not depend on the weights, so we are free to choose them to make $W(G){\cal R}_{st}(G)$ small. By the commute-time identity (\Cref{lem:commute_time}), this is the same as choosing a random walk on $G$ whose commute time between $s$ and $t$ is small, when $s$ and $t$ are connected. The Metropolis-type walk of~\cite{kosowski2013faster}, realized by the Metropolis-Hastings graph $G'$ of \Cref{def:MH}, has $W(G'){\cal R}_{x_sx_t}(G')=O(n^2)$ (\Cref{cor:MH_WR}). In this section we show how to apply it to the setting of $st$-transport in order to obtain a $\widetilde O(n/\eps)$ algorithm.

\subsection{Graph construction}\label{sec:MH-construction}
We extend the Metropolis-Hastings construction to flat connection graphs. Each edge $\{ u,v \}$ is replaced by a path of length $2$. We assign the identity label $I$ to the first edge and the original label $U_{uv}$ to the second, so that traversing the replacement path applies exactly the same unitary as traversing the original edge. This choice preserves transport between vertices corresponding to those of $G$. Below, we formalize the construction and verify that it preserves flatness.
\begin{definition}\label{def:Metropolis-Hastings}
Let $G = (V, E)$ be an unweighted unitary-labeled graph (\defin{labeled-graph}). The corresponding \emph{Metropolis-Hastings unitary-labeled graph} is the weighted Metropolis-Hastings graph $G'=(V',E',w)$ (see \defin{MH}), where,
for each edge $\{u,v\}\in E$ with $u < v$, we define the corresponding unitary labels as $U_{x_u,x_{u,v}} = I, U_{x_{u,v},x_v} = U_{uv}$ (see \fig{MH}). The boundary states $\ket{\psi_s}$ and $\ket{\psi_t}$, together with their state preparation oracles $O_s$ and $O_t$, are inherited unchanged from $G$.
\end{definition}
\begin{figure}
\centering
\includegraphics[width=0.6\textwidth]{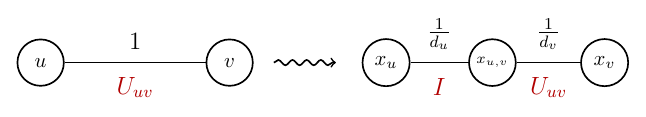}
\caption{The Metropolis-Hastings construction for an edge $\{u,v\}$ such that $u < v$. The original edge of weight $1$ and unitary label $U_{uv}$ is replaced by two edges through a new vertex $x_{u,v}$. The edge $\{x_u, x_{u,v}\}$ has weight $1/d_u$ and unitary label $I$, while the edge $\{x_{u,v}, x_v\}$ has weight $1/d_v$ and unitary label $U_{uv}$.}\label{fig:MH}
\end{figure}
The following lemma shows that $G'$ is flat whenever $G$ is, and that transport between original vertices is unchanged. The proof projects paths in $G'$ onto paths in $G$ by collapsing each split vertex $x_{u,v}$ onto its smaller endpoint $u$, preserving the ordered product of labels.

\begin{lemma}\label{lem:MH_flatness}
If a unitary-labeled graph $G = (V, E)$ satisfies the flatness condition from \defin{labeled-graph}, then the Metropolis-Hastings unitary-labeled graph $G'=(V',E',w)$ of \defin{Metropolis-Hastings} also satisfies the flatness condition, and for all $u,v\in V$ in the same connected component, $U'_{x_u}(x_v)=U_u(v)$.
\end{lemma}
\begin{proof}
Define a map $\pi: V' \to V$ as follows. For every $u \in V$, set
\[
    \pi (x_u) = u.
\]
For every edge $\{u,v\} \in E$ with $u < v$, set
\[
    \pi(x_{u,v}) = u.
\]
Let $\{u,v\}\in E$ with $u<v$. The vertex $x_{u,v}$ is adjacent only to $x_u$ and $x_v$. Thus there are four possible directed steps involving $x_{u,v}$. We use the convention that the label of a step from a vertex to itself is~$I$.

If the step is between $x_u$ and $x_{u,v}$, then both endpoints are mapped to $u$, and the label is $I$ in both directions. This agrees with the label of the corresponding step from $u$ to itself.

If the step is from $x_{u,v}$ to $x_v$, then its label is $U_{uv}$. The corresponding step under the map $\pi$ is from $u$ to $v$ which also has label $U_{uv}$.

If the step is from $x_v$ to $x_{u,v}$ then its label is $U_{vu}$. The corresponding step under the map $\pi$ is from $v$ to $u$ which also has label $U_{vu}$.
Consider any path $y_0, y_1, \ldots, y_L$ in $G'$. Applying the observation above to each consecutive pair $y_\ell, y_{\ell + 1}$ gives
\[
    U_{y_{L - 1} y_L}\cdots U_{y_0 y_1} = U_{\pi(y_{L - 1}) \pi(y_L)} \ldots U_{\pi(y_0) \pi(y_1)}.
\]
By flatness of $G$, the right-hand side depends only on $\pi(y_0)$ and $\pi(y_L)$. Since $\pi$ is fixed, this depends only on $y_0$ and $y_L$. Therefore, the product of labels along a path in $G'$ depends only on its endpoints. Hence $G'$ is flat. Taking $y_0=x_u$ and $y_L=x_v$, the right-hand side is $U_u(v)$, which proves the second claim.
\end{proof}

\subsection{Oracles for the Metropolis-Hastings graph}\label{sec:MH-oracles}

We now show that the oracle access required for the weighted $st$-transport problem on the Metropolis-Hastings graph $G'$ can be implemented efficiently from the corresponding oracle access to the original graph $G$.
The algorithm of \Cref{sec:gauge-state-preparation-transducer} accesses its input graph only through $O_G$ and $O_w$, and through $O_s,O_t$, which are unchanged. So to run it on $G'$, we need to implement $O_{G'}$ and $O_{w'}$. 

\noindent Recall that for every $u\in V$, the neighbours of $u$ in $G$ are indexed by a bijection
\[
    f_u:\{0, \ldots, d_u-1\} \to N(u).
\]
We assume access to the graph oracle for $G$,
\[
    O_G: \ket{u}\ket{i}\ket{\psi} \mapsto \ket{v}\ket{j}U_{uv}\ket{\psi},
\]
where $v=f_u(i)$ and $f_v(j)=u$. Since $G$ is unweighted, the weighted-neighbour oracle $O_w$ of \Cref{prob:gauge_problem_estimate} maps $\ket{u}\ket{0}$ to the uniform superposition over $\{0,\dots,d_u-1\}$. In this section it is convenient to instead work with the degree oracle
\[
    O_{\deg}: \ket{u}\ket{0} \mapsto \ket{u}\ket{d_u},
\]
which can be simulated from $O_G$ alone.

The idea is that $O_G$ acts as the identity on $\ket{u}\ket{i}$ when $i\geq d_u$ is not a valid index, and changes the vertex register when $i<d_u$. So one application of $O_G$ reveals whether $i<d_u$, and binary search then finds $d_u$.
\begin{lemma}\label{lem:degree-oracle}
The degree oracle $O_{\deg}$ can be implemented using $O(\log n)$ calls to $O_G$ and $O(\log^2 n)$ additional elementary gates.
\end{lemma}
\begin{proof}
For a vertex $u$ and an index $i$, we compute the bit $[i<d_u]$ coherently as follows: copy the vertex label $u$ into an ancilla register, apply $O_G$ to $\ket{u}\ket{i}\ket{\psi}$, and record in a fresh qubit whether the vertex register differs from the stored copy. If $i<d_u$, the vertex register contains $f_u(i)\neq u$, since we assume that the graph has no self-loops, whereas otherwise $O_G$ acts as the identity. Apply $O_G^\dagger=O_G$ and uncopy the ancilla to restore all other registers. Each check uses two oracle calls and $O(\log n)$ additional gates.

Since $[i<d_u]$ is monotone in $i$ and $d_u\leq n-1$, a coherent bitwise binary search over $i\in\{0,\ldots,n-1\}$ computes $d_u$ using $O(\log n)$ checks. XOR the result into the output register and reverse the computation to erase the workspace. This uses $O(\log n)$ calls to $O_G$ and $O(\log^2 n)$ additional elementary gates.
\end{proof}
We use the notation $x_{u,v}$ for the split vertex corresponding to the edge $\{u, v\}$ with $u < v$. In the actual encoding, this split vertex is stored as the canonical half edge $(u,i)$ of the vertex with the smaller name, where $f_u(i) = v$. Hence, from the encoding of $x_{u,v}$, we can read off $u$ and $i$, and, using oracle access to $G$, recover $v$ and $j$, where $f_v(j)=u$.
We distinguish the two vertex types by a flag bit, encoding $x_u$ as $(0,u,0)$ and $x_{u,v}$ as $(1,u,i)$, where $u<v$ and $f_u(i)=v$; both encodings use $O(\log n)$ bits.
\begin{lemma}\label{lem:MH_lookup}
The reversible operation
\[
\mathrm{LOOK}:
\ket{u}\ket{i}\ket{0}\ket{0}\ket{\psi}
\mapsto
\ket{u}\ket{i}\ket{v}\ket{j}\ket{\psi},
\]
where $v=f_u(i)$ and $f_v(j)=u$, can be implemented using two calls to $O_G$ and $\widetilde O(1)$ additional gates.
\end{lemma}

\begin{proof}
Apply $O_G$ to the first two registers and to the internal register. This maps $\ket{u}\ket{i}\ket{\psi} \mapsto \ket{v}\ket{j}U_{uv}\ket{\psi}$. Copy the pair $(v,j)$ into the clean registers. Then apply $O_G$ again to the first two registers and to the internal register. Since $U_{vu} = U_{uv}^{\dagger}$, the second query returns the first two registers to $\ket{u} \ket{i}$ and cancels the operation on the internal register. Hence, the final state is $\ket{u}\ket{i}\ket{v}\ket{j}\ket{\psi}$.
\end{proof}
Next, we specify the neighbour indexing of $G'$. For a vertex $x_u$, we use the same neighbour index set as for $u$ in $G$. If $f_u(i) = v$, then the $i$-th neighbour of $x_u$ in $G'$ is the split vertex of the edge $\{u,v\}$, that is, $x_{u,v}$ if $u<v$, and $x_{v,u}$ if $u>v$. For a split vertex $x_{u,v}$ with $u < v$, the two neighbours are indexed by $f'_{x_{u,v}} (0) = x_u$ and $ f'_{x_{u,v}} (1) = x_v$.
Thus, for $u < v$, if $f_u(i)=v$ and $f_v(j)=u$, the edge between $x_u$ and $x_{u,v}$ is indexed by $i$ at $x_u$ and by $0$ at $x_{u,v}$. Similarly, the edge between $x_v$ and $x_{u,v}$ is indexed by $j$ at $x_v$ and by $1$ at $x_{u,v}$.

We can now implement $O_{G'}$. The input is either an original vertex $x_u$ with an index $i$ of $G$, or a split vertex with an index $\ell\in\{0,1\}$. In both cases the procedure is the same: look up the other end of the corresponding edge of $G$, write down the output vertex and index in $G'$, apply the label of the half-edge ($I$ or $U_{uv}$, the latter using one more call to $O_G$), and uncompute the workspace. The only difference between the cases is which output is written and when the label is non-trivial.
\begin{lemma}\label{lem:MH_edge_oracle}
The graph oracle $O_{G'}$ for the Metropolis-Hastings graph $G'$ can be implemented using $O(1)$ calls to $O_G$, $O(1)$ calls to $O_{\deg}$, and $\widetilde O(1)$ additional gates.
\end{lemma}
\begin{proof}
The oracle $O_{G'}$ is required to act as
\[
    O_{G'} : \ket{y}\ket{\ell}\ket{\psi} \mapsto \ket{y'}\ket{\ell'}U_{yy'}\ket{\psi},
\]
where $y' = f'_y(\ell)$ and $f'_{y'}(\ell') = y$, and as the identity on all other inputs.

First, suppose the input is of the form $\ket{x_u}\ket{i}$. We check that $x_u$ is a valid original vertex and that $i<d_u$ using $O(1)$ queries to $O_{\deg}$ and $O_G$, and $\widetilde{O}(1)$ elementary gates. If this check fails, we uncompute and do nothing. If it succeeds, the desired action is
\[
\ket{x_u}\ket{i}\ket{\psi} \mapsto
\begin{cases}
    \ket{x_{u,v}}\ket{0}\ket{\psi} & \text{if } u < v,\\
    \ket{x_{v,u}}\ket{1}U_{uv}\ket{\psi} & \text{if } u > v,
\end{cases}
\]
where $f_u(i) = v$ and $f_v(j) = u$. Indeed, if $u < v$, then the edge from $x_u$ to the split vertex $x_{u,v}$ has label $I$. If $u > v$, then the split vertex is $x_{v,u}$, and the edge from $x_u$ to $x_{v,u}$ has label $U_{uv}$.

Let the clean work registers be initialized to $\ket{0}\ket{0}\ket{0}\ket{0}$, where the first two work registers store the lookup values $(v,j)$ and the last two work registers store the output vertex-index pair. Starting from
\[
\ket{x_u}\ket{i}\ket{0}\ket{0}\ket{0}\ket{0}\ket{\psi},
\]
we implement the map as follows:
\begin{align*}
\ket{x_u}\ket{i}\ket{0}\ket{0}\ket{0}\ket{0}\ket{\psi} &\xmapsto{\;1\;} \ket{x_u}\ket{i}\ket{v}\ket{j}\ket{0}\ket{0}\ket{\psi}\\
&\xmapsto{\;2\;}
\begin{cases}
    \ket{x_u}\ket{i}\ket{v}\ket{j}\ket{x_{u,v}}\ket{0}\ket{\psi} & \text{if } u < v,\\
    \ket{x_u}\ket{i}\ket{v}\ket{j}\ket{x_{v,u}}\ket{1}\ket{\psi} & \text{if } u > v,
\end{cases}\\
&\xmapsto{\;3\;}
\begin{cases}
    \ket{x_u}\ket{i}\ket{v}\ket{j}\ket{x_{u,v}}\ket{0}\ket{\psi} & \text{if } u < v,\\
    \ket{x_u}\ket{i}\ket{v}\ket{j}\ket{x_{v,u}}\ket{1}U_{uv}\ket{\psi} & \text{if } u > v,
\end{cases}\\
&\xmapsto{\;4\;}
\begin{cases}
    \ket{x_{u,v}}\ket{0}\ket{v}\ket{j}\ket{x_u}\ket{i}\ket{\psi} & \text{if } u < v,\\
    \ket{x_{v,u}}\ket{1}\ket{v}\ket{j}\ket{x_u}\ket{i}U_{uv}\ket{\psi} & \text{if } u > v,
\end{cases}\\
&\xmapsto{\;5\;}
\begin{cases}
    \ket{x_{u,v}}\ket{0}\ket{0}\ket{0}\ket{0}\ket{0}\ket{\psi} & \text{if } u < v,\\
    \ket{x_{v,u}}\ket{1}\ket{0}\ket{0}\ket{0}\ket{0}U_{uv}\ket{\psi} & \text{if } u > v.
\end{cases}
\end{align*}
The steps are as follows.
\begin{enumerate}
    \item We use \lem{MH_lookup} to compute $v=f_u(i)$ and $j$ such that $f_v(j) = u$.

    \item Using the registers containing $u,i,v,j$, we compare $u$ and $v$ and write the output vertex-index pair into the clean output registers. If $u < v$, the output pair is $\ket{x_{u,v}}\ket{0}$. If $u > v$, the output pair is $\ket{x_{v,u}}\ket{1}$.

    \item We apply the edge label. If $u < v$, the label is $I$, so we do nothing. If $u > v$, the label from $x_u$ to $x_{v,u}$ is $U_{uv}$, so we apply $U_{uv}$ to the internal register. This is done by applying $O_G$ to the pair $(u,i)$, copied to an auxiliary register, and the internal register, and cleaning the auxiliary register.
    \item We swap the input vertex-index registers with the output vertex-index registers. After this step, the first two registers already contain the desired output.

    \item We clear the work registers by reversibly reconstructing the old input and lookup values from the output pair using \lem{MH_lookup}, XORing these values into the work registers, and then uncomputing the reconstruction.
\end{enumerate}
Thus, ignoring the clean work registers, the action is
\[
\ket{x_u}\ket{i}\ket{\psi} \mapsto
\begin{cases}
    \ket{x_{u,v}}\ket{0}\ket{\psi} & \text{if } u < v,\\
    \ket{x_{v,u}}\ket{1}U_{uv}\ket{\psi} & \text{if } u > v,
\end{cases}
\]
as required.

Next, suppose the input is of the form $\ket{x_{u,v}}\ket{\ell}$, where $u < v$. We check that $\ell \in \{0,1\}$ and that $x_{u,v}$ is a valid split vertex using $O(1)$ queries to $O_{\deg}$ and $O_G$ and $\widetilde{O}(1)$ elementary gates. If this check fails, we uncompute and do nothing. If it succeeds, let $i$ and $j$ be such that $f_u(i) = v$ and $f_v(j) = u$. The desired action is
\[
\ket{x_{u,v}}\ket{\ell}\ket{\psi} \mapsto
\begin{cases}
    \ket{x_u}\ket{i}\ket{\psi} & \text{if } \ell = 0,\\
    \ket{x_v}\ket{j}U_{uv}\ket{\psi} & \text{if } \ell = 1.
\end{cases}
\]
Indeed, the edge from $x_{u,v}$ to $x_u$ has label $I$, while the edge from $x_{u,v}$ to $x_v$ has label $U_{uv}$.

Let the clean work registers be initialized to $\ket{0}\ket{0}\ket{0}\ket{0}$, where the first two work registers store the lookup values $(v,j)$ and the last two work registers store the output vertex-index pair. Starting from
\[
\ket{x_{u,v}}\ket{\ell}\ket{0}\ket{0}\ket{0}\ket{0}\ket{\psi},
\]
we implement the map as follows:
\begin{align*}
\ket{x_{u,v}}\ket{\ell}\ket{0}\ket{0}\ket{0}\ket{0}\ket{\psi} &\xmapsto{\;1\;} \ket{x_{u,v}}\ket{\ell}\ket{v}\ket{j}\ket{0}\ket{0}\ket{\psi}\\
&\xmapsto{\;2\;}
\begin{cases}
    \ket{x_{u,v}}\ket{0}\ket{v}\ket{j}\ket{x_u}\ket{i}\ket{\psi} & \text{if } \ell = 0,\\
    \ket{x_{u,v}}\ket{1}\ket{v}\ket{j}\ket{x_v}\ket{j}\ket{\psi} & \text{if } \ell = 1,
\end{cases}\\
&\xmapsto{\;3\;}
\begin{cases}
    \ket{x_{u,v}}\ket{0}\ket{v}\ket{j}\ket{x_u}\ket{i}\ket{\psi} & \text{if } \ell = 0,\\
    \ket{x_{u,v}}\ket{1}\ket{v}\ket{j}\ket{x_v}\ket{j}U_{uv}\ket{\psi} & \text{if } \ell = 1,
\end{cases}\\
&\xmapsto{\;4\;}
\begin{cases}
    \ket{x_u}\ket{i}\ket{v}\ket{j}\ket{x_{u,v}}\ket{0}\ket{\psi} & \text{if } \ell = 0,\\
    \ket{x_v}\ket{j}\ket{v}\ket{j}\ket{x_{u,v}}\ket{1}U_{uv}\ket{\psi} & \text{if } \ell = 1,
\end{cases}\\
&\xmapsto{\;5\;}
\begin{cases}
    \ket{x_u}\ket{i}\ket{0}\ket{0}\ket{0}\ket{0}\ket{\psi} & \text{if } \ell = 0,\\
    \ket{x_v}\ket{j}\ket{0}\ket{0}\ket{0}\ket{0}U_{uv}\ket{\psi} & \text{if } \ell = 1.
\end{cases}
\end{align*}
The steps are as follows.
\begin{enumerate}
    \item Since the split vertex $x_{u,v}$ is encoded by the canonical half-edge $(u,i)$, we use \lem{MH_lookup} to compute $v = f_u(i)$ and $j$ such that $f_v(j)=u$.

    \item Using the registers containing $u,i,v,j$ and the value of $\ell$, we write the output vertex-index pair into the clean output registers. If $\ell = 0$, the output pair is $\ket{x_u}\ket{i}$. If $\ell = 1$, the output pair is $\ket{x_v}\ket{j}$.

    \item We apply the edge label. If $\ell = 0$, the label is $I$, so we do nothing. If $\ell = 1$, the label from $x_{u,v}$ to $x_v$ is $U_{uv}$, so we apply $U_{uv}$ to the internal register. This is done by applying $O_G$ to a copy of the half-edge $(u,i)$ and the internal register, and cleaning the copy, which then holds $(v,j)$, using the lookup registers.

    \item We swap the input vertex-index registers with the output vertex-index registers. After this step, the first two registers already contain the desired output.

   \item As in the first case, we reversibly reconstruct the old input and lookup values from the output pair using \lem{MH_lookup}, XOR these values into the work registers to clear them, and then uncompute the reconstruction.
\end{enumerate}
Thus, ignoring the clean work registers, the action is
\[
\ket{x_{u,v}}\ket{\ell}\ket{\psi} \mapsto
\begin{cases}
    \ket{x_u}\ket{i}\ket{\psi} & \text{if } \ell = 0,\\
    \ket{x_v}\ket{j}U_{uv}\ket{\psi} & \text{if } \ell = 1,
\end{cases}
\]
as required.

All validity checks are computed coherently into work registers and are uncomputed at the end. If any check fails, the circuit skips the computation of the output pair, skips the label application, skips the swap, and uncomputes the checks, so the operation is the identity on invalid inputs. The number of calls to \lem{MH_lookup} is constant, and each call uses two calls to $O_G$. The label application uses at most one additional call to $O_G$, and the validity checks use $O(1)$ calls to $O_{\deg}$. All other operations are reversible comparisons, controlled copies, controlled swaps, and standard reversible computation on $O(\log n)$-bit registers. Therefore $O_{G'}$ can be implemented using $O(1)$ calls to $O_G$, $O(1)$ calls to $O_{\deg}$, and $\widetilde O(1)$ additional gates.
\end{proof}
The weights of $G'$ were chosen so that $O_{w'}$ is simple. At a non-isolated original vertex $x_u$, all $d_u$ incident edges have the same weight $1/d_u$, so $w'_{x_u}=1$, and $O_{w'}$ only has to prepare a uniform superposition over the indices. A split vertex has just two neighbours, so there $O_{w'}$ is a single-qubit rotation whose angle depends on the degrees of the two endpoints.
\begin{lemma}\label{lem:MH_weight_oracle}
The weighted neighbour oracle $O_{w'}$ for the Metropolis-Hastings graph $G'$ can be implemented using $O(1)$ calls to $O_G$, $O(1)$ calls to $O_{\deg}$, and $\widetilde O(1)$ additional gates, up to the precision used for standard reversible arithmetic and controlled rotations.
\end{lemma}

\begin{proof}
Recall that for every edge $\{u,v\}\in E$ with $u < v$, the weights in $G'$ are
\[
w'_{x_u,x_{u,v}} = \frac{1}{d_u},
\qquad
w'_{x_{u,v},x_v} = \frac{1}{d_v}.
\]

If $d_u = 0$, we let $O_{w'}$ act as the identity at $x_u$. Otherwise, the following construction applies. First, suppose the input vertex is of the form $x_u$. Then the neighbours of $x_u$ in $G'$ are indexed by the same indices as the neighbours of $u$ in $G$. For every $i\in\{0,\ldots,d_u-1\}$, the corresponding edge has weight $1/d_u$. Hence the weighted degree of $x_u$ is
\[
w'_{x_u} = \sum_{i=0}^{d_u-1} \frac{1}{d_u} = 1.
\]
Therefore, the desired action is
\[
O_{w'}\ket{x_u}\ket{0} = \ket{x_u}\frac{1}{\sqrt{d_u}}\sum_{i=0}^{d_u-1}\ket{i}.
\]
This can be implemented by querying $d_u$ using $O_{\deg}$, preparing the uniform superposition over $\{0,\ldots,d_u-1\}$, and then uncomputing the degree register. This uses $O(1)$ calls to $O_{\deg}$ and $\widetilde O(1)$ additional gates.

Now suppose the input vertex is a split vertex $x_{u,v}$, where $u < v$. This vertex has exactly two neighbours, namely $x_u$ and $x_v$. With our indexing,
\[
f'_{x_{u,v}}(0) = x_u,
\qquad
f'_{x_{u,v}}(1) = x_v.
\]
The weighted degree of $x_{u,v}$ is
\[
w'_{x_{u,v}} = \frac{1}{d_u} + \frac{1}{d_v}.
\]
Thus the desired action is
\begin{align*}
O_{w'}\ket{x_{u,v}}\ket{0}
&=
\ket{x_{u,v}}
\left(
\sqrt{\frac{1/d_u}{1/d_u+1/d_v}}\ket{0}
+
\sqrt{\frac{1/d_v}{1/d_u+1/d_v}}\ket{1}
\right)\\
&=
\ket{x_{u,v}}
\left(
\sqrt{\frac{d_v}{d_u+d_v}}\ket{0}
+
\sqrt{\frac{d_u}{d_u+d_v}}\ket{1}
\right).
\end{align*}
To implement this, we query $d_u$ and $d_v$ using $O_{\deg}$, perform the degree-controlled one-qubit rotation
\[
\ket{0} \mapsto \sqrt{\frac{d_v}{d_u+d_v}}\ket{0} + \sqrt{\frac{d_u}{d_u+d_v}}\ket{1},
\]
and uncompute the degree registers.

If the split vertex $x_{u,v}$ is encoded by the canonical half-edge $(u,i)$, where $f_u(i)=v$ and $u<v$, then $u$ and $i$ are already available from the encoding, but $v$ may not be. In that case, we use \lem{MH_lookup} to compute $v=f_u(i)$ and $j$ such that $f_v(j)=u$. We then query $d_u$ and $d_v$, perform the controlled rotation above, and uncompute all work registers. The operation \lem{MH_lookup} uses two calls to $O_G$ and returns the internal register unchanged.

All validity checks are computed coherently and uncomputed at the end. If a check fails, the circuit skips the rotation and acts as the identity. All remaining operations are reversible comparisons, controlled rotations, and standard reversible arithmetic on $O(\log n)$-bit registers. Hence $O_{w'}$ can be implemented using $O(1)$ calls to $O_G$, $O(1)$ calls to $O_{\deg}$, and $\widetilde O(1)$ additional gates, up to the chosen arithmetic precision.
\end{proof}

\begin{corollary}\label{cor:MH_oracles}
Assume access to the graph oracle $O_G$ for the original graph $G$. Then the graph oracle $O_{G'}$ and the weighted neighbour oracle $O_{w'}$ for the Metropolis-Hastings graph $G'$ can both be implemented using $O(\log n)$ calls to $O_G$ and $\widetilde O(1)$ additional gates.
\end{corollary}

\begin{proof}
The statement follows directly from \lem{MH_edge_oracle}, \lem{MH_weight_oracle}, and \lem{degree-oracle}.
\end{proof}

\section[Linear quantum algorithm for st-transport]{Linear quantum algorithm for $st$-transport}\label{sec:linear-alg}

We now combine the results of \Cref{sec:gauge-state-preparation-transducer,sec:MH-connection} to obtain our main upper bound, and then analyze the space complexity of the resulting algorithm.

\begin{theorem}\label{thm:linear-alg}
There is a quantum algorithm that solves \Cref{prob:gauge_problem_estimate} on unweighted graphs with bounded error, meaning that, with probability at least $2/3$, it correctly
reports that $s$ and $t$ are disconnected, or outputs an estimate $\tilde p$ such that
$\abs{\,\abs{\bra{\psi_t}U_s(t)\ket{\psi_s}}^2-\tilde p\,}\leq \eps$, using
$\widetilde{O} \left( n / \eps \right)$ calls to $O_G$, $O_s$ and $O_t$, and
$\widetilde{O}\left( n / \eps \right)$ additional
elementary operations.
\end{theorem}
\begin{proof}
    If $s = t$, we estimate $\abs{\bra{0}O_t^\dagger O_s\ket{0}}^2$ by amplitude estimation using $O(1/\eps)$ oracle calls. Otherwise, if either terminal is isolated, we output
``disconnected''. Henceforth assume $s\ne t$ and $d_s,d_t>0$. Given an input unitary-labeled graph $G$, consider the Metropolis-Hastings unitary-labeled graph $G'$ from \defin{Metropolis-Hastings}.
    By \cor{MH_oracles}, each call to $O_{G'}$ or $O_{w'}$ can be implemented using $O(\log n)$ calls to $O_G$, together with $\widetilde O(1)$ additional elementary operations.
    The boundary states are unchanged by the construction, so the state-preparation oracles $O_s$ and $O_t$ can be used without modification.
    By \lem{MH_flatness}, $G'$ is flat. Moreover, $x_s$ and $x_t$ are connected in $G'$ if and only if $s$ and $t$ are connected in $G$, and in that case
$U'_{x_s}(x_t)\ket{\psi_s}=U_s(t)\ket{\psi_s}$, again by \lem{MH_flatness}.
Consequently, solving $st$-transport on $G'$ with boundary vertices $x_s,x_t$ also solves the original problem on $G$.
The quantities that \Cref{cor:gauge-state-overlap-algorithm} assumes to be known are all determined by $n$: the weighted degrees are $w'_{x_s}=w'_{x_t}=1$, and by \cor{MH_WR} we may take $\mathbf{W}'=2n$ and $\mathbf{R}'=18n$.
Applying \cor{gauge-state-overlap-algorithm} to $G'$ gives a bounded-error quantum algorithm using
$\widetilde O\left( \sqrt{\mathbf W'\mathbf R'} / \eps \right)=\widetilde O(n/\eps)$
calls to $O_{G'}$, $O_{w'}$, $O_s$, and $O_t$, and the same number of additional elementary operations.
Finally, replacing each call to $O_{G'}$ and $O_{w'}$ by the implementations above gives an overall complexity of $\widetilde O(n/\eps)$ calls to $O_G$, $O_s$, and $O_t$, and $\widetilde O(n/\eps)$ additional elementary operations.

The implementation of $O_{w'}$ in \lem{MH_weight_oracle} is exact up to the chosen arithmetic precision -- call this $\delta$. The algorithm makes $T=\widetilde O(n/\eps)$ calls to $O_{w'}$ in total, so a sufficiently precise $\delta\approx \eps/n$ gives $b\approx\log\frac{n}{\eps}$ bits of precision, for a polylog$(n/\eps)=\tO(1)$ overhead.
\end{proof}

\begin{proposition}[Space complexity]\label{prop:linear-alg-space}
For $0<\eps<1$, the algorithm of \Cref{thm:linear-alg} can be implemented using $O\bigl(\log n+\log k+\log(1/\eps)\bigr)$ qubits.
\end{proposition}

\begin{proof}
We use the transducer $S_{\rm fin}$ from \Cref{claim:S-fin}, applied to the Metropolis-Hastings connection graph $G'$. Choose powers of two $M=\Theta(1/\eps)$ and $K=\Theta(n/\eps)$ large enough for amplitude estimation and the conversion of $S_{\rm fin}$ into an algorithm.

Let $Q_{\rm fin}$ be the space needed to implement one controlled call to $S_{\rm fin}$, including all auxiliary registers.
The conversion in \Cref{thm:transducer-implementation} adds $O(\log K)$ qubits \cite[proof of Theorem~5.5 and Lemma~4.6]{belovs2023LasVegasTime}. Thus,
\[
Q_{\rm algo}\leq Q_{\rm fin}+O(\log K),
\]
where $Q_{\rm algo}$ is the total space used by the algorithm. We count the space $Q_{\rm fin}$ in four parts: the clock register $\mathcal S$, the amplitude-estimation register $\mathcal T$, the register for $\mathcal H\oplus\mathcal L_0\oplus\mathcal L_1\oplus\mathcal L$, and the additional workspace needed to implement the transducer as a circuit.

The clock register $\mathcal S$ has dimension $M$ and uses $\log M$ qubits. The amplitude-estimation register $\mathcal T$ also uses $\log M$ qubits. The Fourier transforms act on this same register.

The register for $\mathcal H\oplus\mathcal L_0\oplus\mathcal L_1\oplus\mathcal L$ uses $O(\log n+\log k)$ qubits. Indeed, $\mathcal H\oplus\mathcal L$ is encoded by a vertex of $G'$, a neighbour index, and an internal state in $\mathbb C^k$. These require $O(\log n)$, $O(\log n)$, and $\log k$ qubits, respectively. By \Cref{claim:above}, $\mathcal L_0$ and $\mathcal L_1$ are copies of $\mathcal L$. The dimension of the full direct sum is therefore at most three times that of $\mathcal H\oplus\mathcal L$, so these additional sectors require only $O(1)$ extra qubits.

It remains to bound the workspace of the circuit implementing $S_{\rm fin}$. We use the implementation from \Cref{claim:S-fin}: controlled calls to $U_{\mathcal{AB}}^{\pm1}$, $O_s^{\pm1}$, and $O_t^{\pm1}$, together with Fourier transforms, phase flips, and clock and sector controls. We implement $U_{\mathcal{AB}}$ using the two reflections from \Cref{sec:implementation}, and implement their calls to $O_{G'}$ and $O_{w'}$ using the constructions in \Cref{sec:MH-oracles}.
These implementations use auxiliary registers for intermediate calculations. The auxiliary registers are initialized to zero and returned to zero after use, but must be included in the space bound. We count their space below, starting with the oracle implementations.

The implementation of $O_{G'}$ in \Cref{lem:MH_edge_oracle} uses $O(\log n+\log k+\log(1/\eps))$ qubits, including its argument registers and workspace. It uses a constant number of vertex, index, and degree registers, together with the internal-state register in $\mathbb C^k$.

The implementation of $O_{w'}$ in \Cref{lem:MH_weight_oracle} uses $O(\log n+\log k+\log(1/\eps))$ qubits as well. This includes the internal-state register used by its calls to $O_G$, the degree registers, and the workspace for the degree-controlled rotations and uniform neighbor preparation at the required accuracy.

The uniform-state preparations and degree-controlled rotations use standard reversible arithmetic with $O(\log n+\log(1/\varepsilon))$ workspace, with temporary registers uncomputed and reused. We choose the precision sufficiently high that the accumulated implementation error over $K=O(n/\varepsilon)$ calls is a sufficiently small constant.

The reflections in \Cref{sec:implementation} use these oracles and $O(\log n+\log k)$ further workspace, so they satisfy the same space bound. The construction of $S_{\rm fin}$ in \Cref{claim:S-fin,lem:amp-amp-comp} uses a constant number of calls to these operations. Its Fourier transforms act on $\mathcal T$, and its clock operations and remaining controls use $O(\log M+\log k)$ additional workspace. All temporary registers are uncomputed and reused. Since $M=\Theta(1/\eps)$, the additional workspace is bounded by $O(\log n+\log k+\log(1/\eps))$.
\end{proof}
\section{Lower bound}\label{sec:gauge-lower-bound}

In this section, we show that $st$-transport has a linear quantum query lower bound even under the promise that $s$ and $t$ are connected. We achieve that by reducing the parity problem below to $st$-transport.

\begin{problem}[Parity]\label{prob:parity} Compute $\bigoplus_{i=0}^{n-1}x_i$,
given oracle access to a string $x \in \{0,1\}^n$ via
\[
O_x:\ket{i}\ket{b} \mapsto \ket{i}\ket{b\oplus x_i}.
\]
\end{problem}

\begin{lemma}[\cite{beals2001QLowerBoundPoly,farhi1998parity}]\label{lem:parity_complexity}
The bounded error quantum query complexity of the parity problem (\probl{parity}) is $\Omega(n)$.
\end{lemma}
We establish a linear lower bound for $st$-transport using \lem{parity_complexity}. The following lower bound holds even for instances with $k=2$ in which $s$ and $t$ are promised to be connected.

\begin{figure}
\centering
\includegraphics[width=0.5\textwidth]{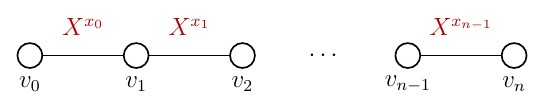}
\caption{The $st$-transport instance used in the lower bound. The graph is the unweighted path $v_0,v_1,\ldots,v_n$, with $s = v_0$ and $t = v_n$. For every $i \in \{0,\ldots,n - 1\}$, the edge $\{ v_i,v_{i+1} \}$ is labeled by the unitary $X^{x_i}$. Thus, transporting the state $\ket{0}$ from $s$ to $t$ applies $X^{\bigoplus_{i = 0}^{n - 1} x_i}$.}\label{fig:lower_bound_new}
\end{figure}
\begin{theorem}\label{thm:gauge-query-lower-bound}
The bounded-error quantum query complexity of \Cref{prob:gauge_problem_estimate}, measured in calls to $O_G$, is $\Omega(n)$ for any $\eps \in (0,1/2)$.
\end{theorem}
\begin{proof}
We reduce the parity problem to $st$-transport. Let $x\in\{0,1\}^n$ be an input to the parity problem. We construct an $st$-transport instance (see \fig{lower_bound_new}) on the unweighted path $G=(V,E)$, where
\[
V = \{v_i : i \in \{0,\ldots,n\}\},
\qquad
E = \{\{v_i,v_{i + 1}\} : i \in \{0,\ldots,n - 1\}\}.
\]
We set $s = v_0,t = v_n$, and assign weight $w_{v_i v_{i+1}} = 1$ to every edge. In particular, $s$ and $t$ are connected for every $x$. We take the internal dimension to be $k = 2$ and set
\[
\ket{\psi_s} = \ket{\psi_t} = \ket{0}.
\]
For every $i \in \{0,\ldots,n-1\}$, define
\[
U_{v_i v_{i+1}} = U_{v_{i+1}v_i} = X^{x_i} .
\]
The flatness condition holds, since the graph contains no cycles. Transporting $\ket{\psi_s}$ from $s$ to $t$ gives
\[
U_s(t)\ket{\psi_s} = X^{x_{n - 1}} \cdots X^{x_1}X^{x_0}\ket{0} = X^{\bigoplus_{i = 0}^{n - 1} x_i} \ket{0}.
\]
Therefore,
\[
\bra{\psi_t}U_s(t)\ket{\psi_s}
=
\begin{cases}
1 & \text{if } \displaystyle\bigoplus_{i=0}^{n-1}x_i=0,\\
0 & \text{if } \displaystyle\bigoplus_{i=0}^{n-1}x_i=1.
\end{cases}
\]
Thus, an estimate of the overlap to additive error strictly smaller than $1/2$ determines the parity of $x$.

It remains to show that the oracle access to this $st$-transport instance can be implemented using a constant number of queries to $O_x$ per query to $O_G$. We use the following local neighbor indexing:
\[
f_{v_0}(0) = v_1,
\qquad
f_{v_n}(0) = v_{n-1},
\]
and, for every $i \in \{1,\ldots,n-1\}$,
\[
f_{v_i}(0) = v_{i-1},
\qquad
f_{v_i}(1) = v_{i+1}.
\]
With this indexing, the graph oracle acts on valid inputs as
\begin{align*}
O_G\ket{v_0}\ket{0}\ket{\psi} &= \ket{v_1}\ket{0}X^{x_0}\ket{\psi},\\
O_G\ket{v_1}\ket{0}\ket{\psi} &= \ket{v_0}\ket{0}X^{x_0}\ket{\psi},\\
O_G\ket{v_i}\ket{0}\ket{\psi} &= \ket{v_{i-1}}\ket{1}X^{x_{i-1}}\ket{\psi} \qquad\text{for }i\in\{2,\ldots,n\},\\
O_G\ket{v_i}\ket{1}\ket{\psi} &= \ket{v_{i+1}}\ket{0}X^{x_i}\ket{\psi} \qquad\text{for }i\in\{1,\ldots,n-1\}.
\end{align*}
To implement this oracle, for every valid pair $(v_i,b)$, where $v_i \in V$ and $b \in \{ 0,1 \}$, we reversibly compute the corresponding edge index, that is the corresponding index in the input string $x$,
\[
x(v_i,b)=
\begin{cases}
0 & \text{if } i = 0, b = 0,\\
i-1 & \text{if } i \neq 0, b = 0,\\
i & \text{if } i \neq 0, i \neq n, b = 1.
\end{cases}
\]
We then apply $O_x$ using the internal qubit as its target register:
\[
O_x \ket{x(v_i,b)}\ket{\psi} = \ket{x(v_i,b)}X^{x_{x(v_i,b)}}\ket{\psi}.
\]
To see the resulting action, write an arbitrary internal state as
\[
\ket{\psi} = \alpha\ket{0} + \beta\ket{1}.
\]
Then, by the definition of $O_x$,
\begin{align*}
O_x\ket{i}\ket{\psi} &= O_x\ket{i}\left(\alpha\ket{0} + \beta\ket{1}\right)\\
&= \alpha\ket{i}\ket{x_i} + \beta\ket{i}\ket{1\oplus x_i}\\
&= \ket{i}X^{x_i}\left(\alpha\ket{0} + \beta\ket{1}\right)\\
&= \ket{i}X^{x_i}\ket{\psi}.
\end{align*}
Thus, one call to $O_x$, with the internal qubit used as its target register, applies exactly the unitary label $X^{x_i}$ to the internal state.
We uncompute the edge index and apply the fixed mapping of the vertex and neighbour-index registers shown above. This implements $O_G$ using one query to $O_x$ and reversible classical computation. Invalid inputs can be detected without querying $O_x$, and the circuit acts as the identity on them.

Next, consider the weighted-neighbour oracle. Since the graph is unweighted,
\[
O_w\ket{v_i}\ket{0} =
\begin{cases}
\ket{v_i}\ket{0}
& \text{if }i \in \{0,n\},\\
\displaystyle
\ket{v_i} \frac{1}{\sqrt{2}} \left( \ket{0} + \ket{1} \right)
& \text{if } i \in \{ 1,\ldots,n - 1 \}.
\end{cases}
\]
The oracle $O_w$ is independent of $x$ and can therefore be implemented without any queries to $O_x$. Similarly, if access to a degree oracle is included, it is independent of $x$, since the endpoints have degree $1$ and all internal vertices have degree $2$. Finally, the boundary state-preparation unitaries can be chosen as $O_s = O_t = I$, and hence they also require no queries to $O_x$.

Consequently, any quantum algorithm that solves $st$-transport using $T$ queries to $O_G$ can be used to solve the parity problem using at most $T$ queries to $O_x$. By \lem{parity_complexity}, this requires $T=\Omega(n)$.
\end{proof}

\bibliographystyle{alpha}
\bibliography{Bibliography}

\end{document}